\documentclass{article}

\usepackage[round]{natbib}
\usepackage{amsmath,amssymb,amsthm,bm,enumerate,mathrsfs,mathtools}
\usepackage{thmtools, thm-restate}
\usepackage{latexsym,color,verbatim,multirow}
\usepackage{graphicx}
\usepackage{caption}
\usepackage{subcaption}
\usepackage{tikz}
\usepackage{geometry}
\usetikzlibrary{shapes,arrows}
\tikzstyle{block} = [rectangle, draw, fill=white!20,
    text width=7em, text centered, rounded corners, minimum height=4em]
\tikzstyle{title} = [text width=7em, text centered, font=\bfseries]
\tikzstyle{line} = [draw, -latex']

\usepackage{mycommands}

\begin{document}

\newtheorem{theorem}{Theorem}
\newtheorem{corollary}[theorem]{Corollary}
\newtheorem{lemma}[theorem]{Lemma}
\newtheorem{observation}[theorem]{Observation}
\newtheorem{proposition}[theorem]{Proposition}
\newtheorem{definition}[theorem]{Definition}
\newtheorem{claim}[theorem]{Claim}
\newtheorem{fact}[theorem]{Fact}
\newtheorem{assumption}[theorem]{Assumption}
\newtheorem{model}[theorem]{Model}

\theoremstyle{definition}
\newtheorem{example}{Example}

\newcommand{\cM}{\mathcal{M}}
\newcommand{\cH}{\mathcal{H}}
\newcommand{\cD}{\mathcal{D}}
\newcommand{\FDR}{\textnormal{FDR}}
\newcommand{\FCR}{\textnormal{FCR}}
\newcommand{\crt}{\phi}
\newcommand{\M}{\mathcal{M}}
\newcommand{\cY}{\mathcal{Y}}
\newcommand{\cX}{\mathcal{X}}
\newcommand{\cV}{\mathcal{V}}
\newcommand{\bX}{\mathbf{X}}
\newcommand{\x}{\mathbf{x}}
\newcommand{\Gv}{\;\;\big|\;\;}
\newcommand{\proj}{\cP}
\newcommand{\pow}{\text{Pow}}
\newcommand{\supp}{\text{supp}}
\newcommand{\sF}{\mathscr{F}}
\newcommand{\cF}{\mathcal{F}}
\newcommand{\sC}{\mathscr{C}}
\newcommand{\hJ}{\widehat{J}}
\newcommand{\bH}{\mathbf{H}}
\newcommand{\bM}{\mathbf{M}}
\newcommand{\tM}{\widetilde{M}}
\newcommand{\tE}{\widetilde{E}}
\newcommand{\tV}{\widetilde{V}}
\newcommand{\tR}{\widetilde{R}}
\newcommand{\tL}{\widetilde{L}}
\newcommand{\hk}{\hat{k}}
\newcommand{\hr}{\hat{r}}       
\newcommand{\cN}{\mathcal{N}}
\newcommand{\cJ}{\mathcal{J}}
\newcommand{\cL}{\mathcal{L}}
\newcommand{\leqAS}{\overset{\textrm{a.s.}}{\leq}}
\newcommand{\Err}{\mathcal{E}}
\newcommand{\RSS}{\text{RSS}}

\newcommand*\mystrut{\vrule width0pt height0pt depth1.5ex\relax}
\newcommand{\underlabel}{\underbracket[1pt][.5pt]{\mystrut \quad\;\; \sub \quad\;\; }}
\newcommand{\JTcomment}[1]{{\color{blue}{(JT: \bf \sc #1) }}}
\newcommand{\WFcomment}[1]{{\color{red}{(WF: \bf \sc #1) }}}
\newcommand{\RTcomment}[1]{{\color{green}{(RT: \bf \sc #1) }}}
\newcommand{\RJTcomment}[1]{{\color{magenta}{(RJT: \bf \sc #1) }}}
\title{Selective Sequential Model Selection}
\author{William Fithian, Jonathan Taylor, Robert Tibshirani, and Ryan J. Tibshirani} 
\maketitle

\begin{abstract}
  Many model selection algorithms produce a path of fits specifying a sequence of increasingly complex models. Given such a sequence and the data used to produce them, we consider the problem of choosing the least complex model that is not falsified by the data. Extending the selected-model tests of \citet{fithian2014optimal}, we construct $p$-values for each step in the path which account for the adaptive selection of the model path using the data. In the case of linear regression, we propose two specific tests, the max-$t$ test for forward stepwise regression (generalizing a proposal of \citet{buja2014}), and the next-entry test for the lasso. These tests improve on the power of the saturated-model test of \citet{tibshirani2014exact}, sometimes dramatically. In addition, our framework extends beyond linear regression to a much more general class of parametric and nonparametric model selection problems.

To select a model, we can feed our single-step $p$-values as inputs into sequential stopping rules such as those proposed by \citet{gsell2013sequential} and \citet{li2015accumulation}, achieving control of the familywise error rate or false discovery rate (FDR) as desired. The FDR-controlling rules require the null $p$-values to be independent of each other and of the non-null $p$-values, a condition not satisfied by the saturated-model $p$-values of \citet{tibshirani2014exact}. We derive intuitive and general sufficient conditions for independence, and show that our proposed constructions yield independent $p$-values.
\end{abstract}

\section{Introduction}
\label{sec:introduction}
Many model selection procedures produce a sequence of increasingly complex models, leaving the data analyst to choose among them. 
Given such a path, we consider the problem of choosing the simplest model in the path that is not falsified by the available data. Examples of path algorithms include forward stepwise linear regression, least angle regression (LAR) and the ever-active path in lasso ($\ell_1$-regularized) regression. \citet{tibshirani2014exact} study methods for generating exact $p$-values at each step of these path algorithms, and their methods provide a starting point for our proposals. Other related works include \citet{loftus2014significance}, who describe $p$-values for path algorithms that add groups of variables (instead of individual variables) at each step, and \citet{choi2014selecting}, who describe $p$-values for steps of principle components analysis (each step marking the estimation of a principle component direction).

We consider the following workflow: to select a model, we compute a set of sequential $p$-values at each step of the model path, and then feed them into a stopping rule that is guaranteed to control the false discovery rate (FDR) \citep{benjamini1995controlling}, familywise error rate (FWER), or a similar quantity. Recently \citet{gsell2013sequential} and \citet{li2015accumulation} proposed sequential stopping rules of this kind. Both sets of rules require that once we have reached the first correct model in our path, the $p$-values in subsequent steps are uniform and independent. While this is not true of the $p$-values constructed in \citet{tibshirani2014exact}, in this paper we develop a theoretical framework for constructing sequential $p$-values satisfying these properties, and give explicit constructions.

Our approach and analysis are quite general, but we begin by introducing a specific example: the selective max-$t$ test for forward stepwise regression.  This is a selective sequential version of the max-$t$ test of \citet{buja2014}, generalized using the theory in \citet{fithian2014optimal}.

\subsection{The max-$t$ Test in Forward Stepwise Regression}

Forward stepwise regression is a greedy algorithm for building a sequence of nested linear regression models. At each iteration, it augments the current model by including the variable that will minimize the residual sum of squares (RSS) of the next fitted model. Equivalently, it selects the variable with the largest $t$-statistic, adjusting for the variables in the current model.

For a design matrix $X\in \R^{n\times p}$, with columns $X_1,\ldots X_p \in \R^n$, and a response vector $Y\in \R^n$, let $E \sub \{1, \ldots, p\}$ denote the current {\em active set}, the set of predictor variables already selected, and let $\RSS(E)$ denote the residual sum of squares for the corresponding regression model. For $j\notin E$ let $t_{j,E}(Y)$ denote the multivariate $t$-statistic of variable $j$, adjusting for the active variables. Using the standard result that
\begin{equation*}
t_{j,E}^2 = (n-|E|-1) \frac{\RSS(E) - \RSS(E \cup \{j\})}{\RSS(E \cup \{j\})},
\end{equation*}
we see that the next variable selected is 
\begin{equation*}
j^* = \argmax_{j\notin E} \, |t_{j,E}| = \argmin_{j\notin E} \, \RSS(E \cup \{j\}), 
\end{equation*}
with corresponding $t$-statistic $t_E^* = t_{j^*,E}$. The selective max-$t$ test rejects for large values of $|t_E^*|$, compared to an appropriate conditional null distribution that accounts for the adaptive selection of the model path.

Table~\ref{tab:diab} illustrates the max-$t$ test and two others applied to the diabetes data from \cite{LARS}, consisting of observations of $n=442$ patients. The response of interest is a quantitative measure of disease progression one year after baseline, and there are ten measured predictors --- age, sex, body-mass index, average blood pressure, and six blood serum measurements --- plus quadratic terms, giving a total of $p=64$ features. We apply forward stepwise regression to generate the model path (beginning with an intercept term, which we represent as a predictor $X_0=(1,\ldots,1)'$), and then use each of the three methods to generate $p$-values at each step along the path. Finally, we use the ForwardStop rule \citep{gsell2013sequential} at FDR level $\alpha=0.1$ to select a model based on the sequence of $p$-values. The bolded entry in each column indicates the last variable selected by ForwardStop.

\begin{table}[ht]
\centering
{\small \begin{tabular}{|l|c|c|c|c|c|}
\hline
 Step & Variable &  Nominal $p$-value &  Saturated $p$-value &  Max-$t$ $p$-value \\
\hline
    1 &      bmi &            0.00 &              0.00 &         0.00 \\
    2 &      ltg &            0.00 &              0.00 &         0.00 \\
    3 &      map &            0.00 &             {\bf 0.05} &         0.00 \\
    4 &  age:sex &            0.00 &              0.33 &         0.02 \\
    5 &  bmi:map &            0.00 &              0.76 &         0.08 \\
    6 &      hdl &            0.00 &              0.25 &         0.06 \\
    7 &      sex &            0.00 &              0.00 &         0.00 \\
    8 &    glu$^2$ &            0.02 &              0.03 &         {\bf 0.32} \\
    9 &    age$^2$ &            0.11 &              0.55 &         0.94 \\
   10 &  map:glu &            0.17 &              0.91 &         0.91 \\
   11 &       tc &            0.15 &              0.37 &         0.25 \\
   12 &      ldl &            0.06 &              0.15 &         0.01 \\
   13 &    ltg$^2$ &            0.00 &              0.07 &         0.04 \\
   14 &  age:ldl &            0.19 &              0.97 &         0.85 \\
   15 &   age:tc &            0.08 &              0.15 &         0.03 \\
   16 &  sex:map &            0.18 &              0.05 &         0.40 \\
   17 &      glu &            0.23 &              0.45 &         0.58 \\
   18 &      tch &       {\bf     0.31} &              0.71 &         0.82 \\
   19 &  sex:tch &            0.22 &              0.40 &         0.51 \\
   20 &  sex:bmi &            0.27 &              0.60 &         0.44 \\
\hline
\end{tabular}

}
\caption[tab:diab]{Forward stepwise regression for the diabetes data: naive $p$-values, $p$-values from the saturated model, and our max-$t$ $p$-values. The bold annotation in each column marks the model selected by ForwardStop with $\alpha=0.1$. For the max-$t$ $p$-values, this stopping rule gives exact FDR control at the 10\% level.}
\label{tab:diab}
\end{table}

The nominal (unadjusted) $p$-values in column 1 of Table~\ref{tab:diab} are computed by comparing $t_E^*$ to the $t$-distribution with $n-|E|-2$ degrees of freedom, which would be the correct distribution if the sequence of models were selected before observing the data. Because the model sequence is in fact selected adaptively to maximize the absolute value of $t_E^*$, this method is highly anti-conservative.

The max-$t$ test $p$-values in column 3 are computed using the same test statistic $|t_E^*|$, but compared with a more appropriate null distribution. We can simulate the distribution of $|t_E^*|$ under the null model, i.e., the linear regression model specified by active set $E$:
\begin{equation*}
M(E):\; Y \sim \cN(X_E\beta_E, \sigma^2 I_n), \quad \beta_E\in \R^{|E|}, \sigma^2>0,
\end{equation*}
where $X_E$ is the matrix with columns $(X_j)_{j\in E}$. As $U_E=\left(X_E'Y,\, \|Y\|^2\right)$ is the complete sufficient statistic for $M(E)$, we can sample from the conditional null distribution of $t_E^*$ given $U_E$ and the current active set $E$, using the {\em selected-model test} framework of \citet{fithian2014optimal}. In step one, because $E=\{0\}$ is fixed and $t_E^*$ is independent of $U_{\{0\}}=(\sum_i Y_i, \|Y\|^2)$, the conditional test reduces to the max-$t$ test proposed in \citet{buja2014}. In later steps, $t_E^*$ and $U_E$ are not conditionally independent given $E$, but we can numerically approximate the null conditional distribution through Monte Carlo sampling.

The saturated-model $p$-values in column 2 also account for selection, but they rely on somewhat different assumptions, condition on more information, and test a slightly different null hypothesis. We discuss these distinctions in detail in Section~\ref{sec:whichnull}.

For most of the early steps in Table~\ref{tab:diab}, the (invalid) nominal test gives the smallest $p$-values, followed by the max-$t$ test, then the saturated-model test. Both the max-$t$ and saturated-model $p$-values are exactly uniform under the null, but the max-$t$ $p$-values appear to be more powerful under the alternative. As we will discuss in Section~\ref{sec:selective-reg}, selected-model tests such as the max-$t$ can be much more powerful than saturated-model tests in early steps of the path when multiple strong variables are competing with each other to enter the model first. In the diabetes example, using the max-$t$ $p$-values, ForwardStop selects a model of size 8, compared to size 3 when using the saturated-model $p$-values. 

The null max-$t$ $p$-values are independent, while the saturated-model $p$-values are not, and hence ForwardStop is guaranteed to control FDR when using the former but not the latter. In Section~\ref{sec:pValsIndep} we discuss intuitive sufficient conditions for independence and show that the max-$t$ and many other selected-model tests satisfy these conditions.

\subsection{Outline}

For the remainder of the article we will discuss the problem of selective sequential model selection in some generality, returning periodically to the selective max-$t$ test in forward stepwise regression and its lasso regression counterpart, the next-entry test, as intuitive and practically important running examples. 

In comparison to their saturated-model counterparts, we will see that the selected-model tests proposed here have three main advantages: power, independence (in many cases), and generalizability beyond linear regression. These advantages do not come entirely for free, as we will see: when $\sigma^2$ is known, the selected-model methods test a more restrictive null hypothesis than the saturated-model methods. In addition, the selected-model tests require accept-reject or Markov Chain Monte Carlo (MCMC) sampling, while the saturated-model tests can be carried out in closed form.

Sections~\ref{sec:genericSetting} and \ref{sec:ordered} introduce the general problem setting and review relevant literature on testing ordered hypotheses using a sequence of $p$-values. Many of these methods require the $p$-values to be uniform and independent under the null, and we derive general conditions to ensure this property in Section~\ref{sec:pValsIndep}. In Section~\ref{sec:selective-reg} we contrast selected-model and saturated-model tests in the linear regression setting, and explain why selected-model tests are often much more powerful in early steps of the path. In Section~\ref{sec:computation} we discuss strategies for sampling from the null distribution for our tests, and prove that in forward stepwise regression and lasso ($\ell_1$-regularized) regression, the number of constraints in the sampling problem is never more than twice the number of variables. We provide a simulation study in Section~\ref{sec:sparseReg}, and Section~\ref{sec:further-examples} discusses applications of our framework beyond linear regression to decision trees and nonparametric changepoint detection. The paper ends with a discussion.

\section{Selective Sequential Model Selection}
\label{sec:genericSetting}

In the general setting, we observe data $Y \in \cY$ with unknown sampling distribution $F$, and then apply some algorithm to generate an adaptive sequence of $d$ nested statistical models contained in some {\em upper model} $M_{\infty}$:
\[
M_0(Y) \sub M_1(Y) \sub \cdots \sub M_d(Y) \sub M_\infty.
\]
By {\em model}, we mean a family of candidate probability distributions for $Y$. For example in  linear regression, a model specifies  the set of
predictors allowed to have nonzero coefficients (but not the  values of their coefficients). Note the ``assumption'' of an upper model involves no loss of generality --- we can always take $M_\infty$ as the union of all models under consideration. We will use the notation $M_{0:d}$ to denote the sequence $M_0, \ldots, M_d$, which we call the {\em model path}. We assume throughout that for each candidate model $M$ there is a minimal sufficient statistic $U(Y; \,M)$, and write
\[
U_k(Y) = U(Y; \,M_k(Y)).
\]
Given data $Y$ and model path $M_{0:d}$, we set ourselves the formal goal of selecting the simplest correct model: the smallest $k$ for which $F\in M_k(Y)$. Of course, in most real data problems all of the $M_k(Y)$ in the path, as well as all other models under consideration, are incorrect. In more informal terms, then, we seek the simplest model in the path that is not refuted by the available data.

Define the {\em completion index} by $k_0(Y, F) = \min\{k:\; F \in M_k(Y)\}$, the index of the first correct model. By construction, $F\in M_k \iff k \geq k_0$. A {\em stopping rule} is any estimator \smash{$\hk$} of $k_0$, with $M_k$ ``rejected'' if $k < \hk$, and ``accepted'' otherwise. Because \smash{$\hk$} is the number of models we do reject, and $k_0$ is the number we should reject, the number of type I errors is \smash{$V=(\hk-k_0)_+$}, while the number of type II errors is \smash{$(k_0-\hk)_+$}. Depending on the scientific context, we might want to control the FWER: $\P(V>0)$, the FDR: \smash{$\E[V/\hk; \hk>0]$}, or another error rate such as a modified FDR, defined by the expectation of some loss function \smash{$g(\hk, k_0)$}:
\begin{equation}\label{eq:errRate}
\Err_F(\hk(\cdot), g) = \E_F\left[ g\left(\hk(Y), k_0(Y, F)\right)\right].
\end{equation}

\subsection{Single-Step $p$-Values}\label{sec:singleStep}

We will restrict our attention to stopping rules like those recently proposed by \citet{gsell2013sequential} and \citet{li2015accumulation}, which operate on a sequence $p_{1:d}$ of $p$-values. At each step $k$, we will construct a $p$-value for testing
\[
H_{k}:\; F\in M_{k-1}(Y)
\]
against the alternative that $F\in M_\infty\setminus M_{k-1}$, accounting for the fact that $M_{k-1}$ is chosen adaptively.

Following \citet{fithian2014optimal}, we say that for a fixed candidate null model $M$, the random variable $p_{k,M}(Y)$ is a valid {\em selective $p$-value} for $M$ at step $k$ if it is stochastically larger than uniform (henceforth super-uniform) under sampling from any $F\in M$, given that $M$ is selected. That is,
\begin{equation*}
\P_F\left(p_{k,M}(Y) \leq \alpha \mid M_{k-1}(Y) = M\right) 
\leq \alpha, \quad \forall F\in M, \; \alpha \in [0,1].
\end{equation*}
Once we have constructed selective $p$-values for each fixed $M$, we can use them as building blocks to construct a combined $p$-value for the random null $M_{k-1}(Y)$. Define
\[
p_k(y) = p_{k, M_{k-1}(y)}(y),
\]
which is super-uniform on the event $\{F \in M_{k-1}(Y)\}$:
\begin{equation}\label{eq:selectiveGuaranteePk}
\P_F\left(p_k \leq \alpha \mid F\in M_{k-1}\right) \leq \alpha, \quad \forall \alpha \in [0,1].
\end{equation}
One useful strategy for constructing exact selective tests is to condition on the sufficient statistics of the null model $M_{k-1}$. By sufficiency, the law
\[
\cL_F(Y \;\mid\; U_{k-1}, \; M_{k-1}=M)
\]
is the same for every $F\in M$. Thus, we can construct selective tests and $p$-values by comparing the realized value of any test statistic $T_k(Y)$ to its known conditional distribution under the null. It remains only to choose the test statistic and compute its conditional null distribution, which can be challenging. See \citet{fithian2014optimal} for a general treatment. Section~\ref{sec:computation} discusses computational strategies for the tests we propose in this article.

\subsection{Sparse Parametric Models}\label{sec:genSparse}

Many familiar path algorithms, including forward stepwise regression, least angle regression (LAR), and the lasso regression (thought of as producing a path of models over its regularization parameter $\lambda$), are methods for adaptively selecting a set of predictors in linear regression models where we observe a random response $Y\in \R^n$ as well as a fixed design matrix $X\in \R^{n \times p}$, whose columns correspond to candidate predictors. For each possible active set $E \sub \{1,\ldots,p\}$, there is a corresponding candidate model
\[
M(E):\; Y \sim \cN( X_E\beta_E, \sigma^2 I_n),
\]
which is a subset of the {\em full model}
\[
M_\infty:\; Y \sim \cN(X\beta, \sigma^2I_n).
\]
If the error variance $\sigma^2$ is known, the complete sufficient statistic for $M(E)$ is $X_E'Y$; otherwise it is $\left(X_E'Y,\, \|Y\|^2\right)$.

For the most part, we can discuss our theory and proposals in a parametric setting generalizing the linear regression problem above. Let $M_\infty$ be a model parameterized by $\theta\in \Theta \sub \R^J$:
\[
M_\infty = \{F_\theta:\; \theta \in \Theta\}.
\]
For any subset $E\sub \{1,\ldots, J\}$ define the sparse submodel with active set $E$ as follows:
\[
\Theta(E) = \{\theta:\; \theta_j = 0, \;\;\forall j \notin E\}, 
\quad M(E) = \{F_\theta:\; \theta\in \Theta(E)\}.
\]
We can consider path algorithms that return a sequence of nested active sets
\[
E_0(Y) \sub E_1(Y) \sub \cdots \sub E_d(Y) \sub \{1,\ldots J\},
\]
inducing a model path with $M_k = M(E_k),$ for $k=0,\ldots,d$. We will be especially interested in two generic path algorithms for the sparse parametric setting: forward stepwise paths and ever-active regularization paths. As we will see in Section~\ref{sec:indep-greedy-entry}, both methods generically result in independent null $p$-values. For a nonparametric example see Section~\ref{sec:nonpar}.

\subsubsection{Forward Stepwise Paths and Greedy Likelihood Ratio Tests}
Let $\ell(\theta; Y)$ denote the log-likelihood for model
$M_\infty$. A generic {\em forward stepwise} algorithm proceeds as follows: we begin with some fixed set $E_0$ (such as an intercept-only model), then at step $k=1,\ldots,d$, we set
\begin{equation}
j_k = \argmax_j \; \sup \left\{\ell(\theta; Y):\; \theta\in\Theta(E_{k-1} \cup \{j\})\right\}, \quad \text{ and } \quad
E_k = E_{k-1} \cup \{j_k\}.
\end{equation}
That is, at each step we select the next variable to maximize the likelihood of the next fitted model (in the case of ties, we could either choose randomly or select both variables).

A natural choice of test statistic is the {\em greedy likelihood ratio statistic}
\begin{equation}\label{eq:greedyLRT}
G_k(Y) = \sup_{\theta\in \Theta(E_k)} {\ell(\theta; Y)} \;\;- \sup_{\theta\in \Theta(E_{k-1})} {\ell(\theta; Y)},
\end{equation}
which is the generalized likelihood ratio statistic for testing $M(E_{k-1})$ against the ``greedy'' alternative with one more active parameter, $\bigcup_{j \notin E_{k-1}} M(E_{k-1} \cup \{j\})$. The {\em selective greedy likelihood ratio test} rejects for large $G_k$, based on the law
\begin{equation*}
\cL\left(G_k \;\mid\; M_{k-1}, U_{k-1}\right)
\end{equation*}
Because the likelihood in linear regression is a monotone decreasing function of the residual sum of squares, the max-$t$ test is equivalent to the greedy likelihood ratio test. The counterpart of the max-$t$ test in linear regression with known $\sigma^2$ is the max-$z$ test, which differs only in replacing the $t$-statistics $t_{j,E}$ with their corresponding $z$-statistics. The max-$z$ test is also equivalent to the selective greedy likelihood ratio test.

For simplicity, we have implicitly made two assumptions: that only one variable is added at each step, and that the set of candidate variables we choose from is the same in each step. It is relatively straightforward to relax either assumption, but we do not pursue such generalizations here.

\subsubsection{Ever-Active Regularization Paths and Next-Entry Tests}
Another important class of model selection procedures is the sequence of {\em ever-active} sets for a regularized likelihood path, under a sparsity-inducing regularizer such as a scaled $\ell_1$ norm. The notion of ever-active sets is needed since these solution paths can drop  (as well as add) predictors along the way. For some ordered set $\Lambda$, let $\lambda\in\Lambda$ parameterize the regularization penalty $P_\lambda(\theta)$. As an example, for a lasso penalty, this is $P_\lambda = \lambda\|\beta\|_1$ with $\lambda\in (0,\infty)$. 

Assume for simplicity that there is a unique solution to each penalized problem of the form
\begin{equation}\label{eq:pen-lik}
  \hat\theta^{\lambda}(Y) =
  \argmin_{\theta\in\Theta} -\ell(\theta; Y) + P_\lambda(\theta).
\end{equation}
It may be impossible or inconvenient to compute \smash{$\hat\theta^\lambda(Y)$} for every $\lambda\in(0,\infty)$. If so, we can instead take $\Lambda$ to be a grid of finitely many values. 

We define the {\em ever-active set} for $\lambda\in\Lambda$ as
\begin{equation}
  \tE_\lambda(Y) = \left\{j:\; \hat\theta_j^\gamma(Y) \neq 0 
    \text{ for any } \gamma\in \Lambda, \gamma \geq \lambda \right\}
\end{equation}
Note that the ever-active sets \smash{$\tE_\lambda$} are nested by construction. In addition, we will assume \smash{$|\tE_\lambda|<\infty$} for every $\lambda\in\Lambda$. Our model path will correspond to the sequence of distinct ever-active sets. Formally, let 
\[
\lambda_0=\sup \Lambda, \quad \text{ and } \quad 
E_0 = \bigcap_{\lambda\in \Lambda} \tE_\lambda,
\]
and for $k\geq 1$, let $\lambda_k$ denote the (random) value of $\lambda$ where the active set changes for the $k$th time:
\begin{equation*}
  \Lambda_k = \{\lambda\in \Lambda:\; \tE_\lambda \supsetneq \tE_{\lambda_{k-1}}\},
  \quad
  \lambda_k = \sup\Lambda_k,
  \quad \text{ and } \quad
  E_k = \bigcap_{\lambda\in\Lambda_k} \tE_{\lambda}.
\end{equation*}

In this setting, $\lambda_k$ is a natural test statistic for model $M_{k-1}$, with larger values suggesting a poorer fit. The {\em selective next-entry test} is the test that rejects for large $\lambda_k$, based on the law
\begin{equation*}
\cL\left(\lambda_k \;\mid\; U_{k-1}, M_{k-1}\right).
\end{equation*}

\paragraph{Remark.} In its usual formulation, the lasso coefficients for linear regression minimize a penalized RSS criterion. If we replace $-\ell$ with any strictly decreasing function of the log-likelihood such as RSS, all of the results in this article hold without modification.

\subsection{Which Null Hypothesis?}
\label{sec:whichnull}

Note that in our formulation of the problem, the type I error \smash{$V=(\hk-k_0)_+$} is defined in a ``model-centric'' fashion: at step $k$ in linear regression, we are testing  whether a particular linear model $M(E_{k-1})$ adequately describes the data $Y$. Even if the next selected variable $X_{j_k}$ has a zero coefficient in the full model, it is not a mistake to reject $M(E_{k-1})$ provided there are some signal variables that have not yet been included.

Depending on the scientific context, we might want to define a type I error at step $k$ differently, by choosing a different null hypothesis to test. Let $\mu = \E Y$ and let $\theta^E$ denote the {\em least-squares coefficients} of active set $E$ --- the coefficients of the best linear predictor for the design matrix $X_E$:
\[
\theta^E = X_E^+ \mu = \argmin_{\theta\in \R^{|E|}} \; \|\mu - X_E\theta\|_2^2,
\]
where $X_E^+$ is the Moore-Penrose pseudoinverse of the matrix $X_E$.

\citet{gsell2013sequential} describe three different null hypotheses that we could consider testing at step $k$ in the case of linear regression:
\begin{description}
\item[Complete Null:] $M_{k-1}$ is (already) correct. That is, 
\[
H_k:\;\mu = X_{E_{k-1}} \theta^{E_{k-1}}.
\]
\item[Incremental Null:] $M_{k-1}$ may be incorrect, but $M_k$ is no improvement. That is, 
\[
H_k^{\text{inc}}:\; \theta_{j_k}^{E_k} = 0.
\]
\item[Full-Model Null:] The coefficient of $X_{j_k}$ is zero in the ``full'' model with all $p$ predictors. That is,
\[
H_k^{\text{full}}:\; \theta_{j_k}^{\{1,\ldots,p\}} = 0.
\]
\end{description}

While the complete null is the strongest null hypothesis of the three, the incremental null is neither weaker nor stronger than the full-model null. Defining 
\begin{align*}
V^{\text{inc}} &= \#\{k < \hk:\; H_k^{\text{inc}} \text{ is true}\}, \quad \text{ and } \\
V^{\text{full}} &= \#\{k < \hk:\; H_k^{\text{full}} \text{ is true}\},
\end{align*}
we can define an analogous FWER and FDR with respect to each of these alternative choices, and attempt to control these error rates. For example, we could define
\[
\text{FDR}^{\text{full}} = \E[V^{\text{full}} / (\hk \vee 1)],
\]
as the false discovery rate with respect to the full-model null. \citet{barber2014controlling} present a framework for controlling $\text{FDR}^{\text{full}}$.

The full-model null is the most different conceptually from the other two, taking a ``variable-centric'' instead of ``model-centric'' viewpoint, with the focus on discovering variables that have nonzero coefficients after adjusting for all $p-1$ other variables under consideration. To elucidate this distinction, consider a bivariate regression example in which the two predictors $X_1$ and $X_2$ are both highly correlated with $Y$, but are also nearly collinear with each other, making it impossible to distinguish which variable has the ``true'' effect. Any procedure that controls $\text{FWER}^{\text{full}}$ could not select either variable, and would return the empty set of predictors. By contrast, most of the methods presented in this article would select the first variable to enter (rejecting the global null model), and then stop at the univariate model. 

Similarly, consider a genomics model with quantitative phenotype $Y$ and predictors $X_j$, representing minor allele counts for each of $p$ single-nucleotide polymorphisms (SNPs). If correlation between neighboring SNPs (neighboring $X_j$'s) is high, it may be very difficult to identify SNPs that are highly correlated with $Y$, adjusting for all other SNPs; however, we might nevertheless be glad to select a model with a single SNP from each important {\em gene}, even if we cannot guarantee it is truly the ``best'' SNP from that gene.

As the above examples illustrate, methods that control full-model error rates are best viewed not as {\em model-selection} procedures --- since all inferences are made with respect to the full model --- but instead as {\em variable-selection} procedures that test multiple hypotheses with respect to a single model, which is specified ahead of time. The ``model'' returned by such procedures is not selected or validated in any meaningful sense. Indeed, in the bivariate example above, of the four models under consideration ($\emptyset$, $\{1\}$, $\{2\}$, and $\{1,2\}$), only the global null model is clearly inconsistent with the data; and yet, a full-model procedure is bound not to return any predictors.

Because the truth or falsehood of full-model hypotheses can depend sensitively on the set of $p$ predictors, rejecting $H_j^{\text{full}}$ has no meaning without reference to the list of all variables that we controlled for. As a result, rejections may be difficult to interpret when $p$ is large. Thus, error rates like $\text{FDR}^{\text{full}}$ are best motivated when the full model has some special scientific status. For example, the scientist may believe, due to theoretical considerations, that the linear model in $X_1,\ldots,X_p$ is fairly credible, and that a nonzero coefficient $\beta_j$, controlling for all of the other variables, would constitute evidence for a causal effect of $X_j$ on the response. 

In this article we will concern ourselves primarily with testing the complete null, reflecting our stated aim of {\em choosing the least complex model that is consistent with the data}. As we discuss further in Section~\ref{sec:selective-reg}, the saturated-model tests of \citet{tibshirani2014exact} are valid selective tests of $H_k^{\text{inc}}$. The advantage of these tests is that they are highly computationally efficient (they do not require sampling). But, unfortunately, they also carry a number of drawbacks: they require us to assume $\sigma^2$ is known, can result in a large reduction in power, create dependence between $p$-values at different steps, and are difficult to generalize beyond the case of linear regression.

\subsection{Related Work}

The problem of model selection is an old one, with quite an extensive literature. However, with the exception of the works above, very few methods offer finite-sample guarantees on the model that is selected except in the orthogonal-design case. One  exception is the knockoff filter of \citet{barber2014controlling}, a variant of which provably controls the full-model FDR.
We compare our proposal to the knockoff method in Section \ref{sec:sparseReg}.

Methods like AIC \citep{akaike1974new} and BIC \citep{schwarz1978estimating} are designed for the non-adaptive case, where the sequence of models is determined in advance of observing the data. Cross-validation, another general-purpose algorithm for tuning parameter selection, targets out-of-sample error and tends to select many noise variables when the signal is sparse \citep[e.g.][]{LY2015}. \citet{benjamini2009simple} extend the AIC by using an adaptive penalty $\sigma^2\lambda_{k,p} k$ to select a model, based on generalizing the Benjamini-Hochberg procedure \citep{benjamini1995controlling}, but do not prove finite-sample control. The stability selection approach of \citet{meinshausen2010stability} uses an approach based on splitting the data many times and offers asymptotic control of the FDR, but no finite-sample guarantees are available.

If the full model is sparse, and the predictors are not too highly correlated, it may be possible asymptotically to recover the support of the full-model coefficients with high probability --- the property of {\em sparsistency}. Much recent model-selection literature focuses on characterizing the regime in which sparsistency is possible; see, e.g. \citet{bickel2009simultaneous}, \citet{meinshausen2006high}, \citet{negahban2009unified}, \citet{van2009conditions}, \citet{wainwright2009sharp}, \citet{zhao2006model}. Under this regime, there is no need to distinguish between the ``model-centric'' and ``variable-centric'' viewpoints. However, the required conditions for sparsistency can be difficult to verify, and in many applied settings they fail to hold. By contrast, the methods presented here require no assumptions about sparsity or about the design matrix $X$, and offer finite-sample guarantees.

\section{Stopping Rules and Ordered Testing}\label{sec:ordered}

An ordered hypothesis testing procedure takes in a sequence of $p$-values $p_1, \ldots, p_d$ for null hypotheses $H_{1}, \ldots, H_{d}$, and outputs a decision \smash{$\hk$} to reject the initial block \smash{$H_{1}, \ldots, H_{\hk}$} and accept the remaining hypotheses. Note that in our setup the hypotheses are nested, with $H_{k-1} \sub H_k$; as a result, all of the false hypotheses precede all of the true ones.

We first review several proposals for ordered-testing procedures, several of which require independence of the null $p$-values conditional on the non-null ones. These procedures also assume the sequence of hypotheses is fixed, whereas in our setting the truth or falsehood of $H_k$ is random, depending on which model is selected at step $k$. In Section~\ref{sec:random-hyp} we show that the error guarantees for these stopping rules do transfer to the random-hypothesis setting, provided we have the same independence property conditional on the completion index $k_0$ (recall 
$k_0(Y, F) = \min\{k:\; F \in M_k(Y)\}$).

\subsection{Proposals for Ordered Testing of Fixed Hypotheses}
\label{sec:orderedProposals}

We now review several proposals for ordered hypotheses testing along with their error guarantees in the traditional setting, where the sequence of null hypotheses is fixed. In the next section we will extend the analysis to random null hypotheses.

The simplest procedure is to keep rejecting until the first time that $p_k > \alpha$: 
\[
\hk_B(Y) = \min\left\{k:\; p_k > \alpha\right\} - 1
\]
We will call this procedure {\em BasicStop}. It is discussed in \citet{marcus1976}.
Since $k_0$ is the index of the first null hypothesis, we have
\begin{equation}\label{eq:basic-stop-control}
 \text{FWER} \leq \P(p_{k_0} \leq \alpha) \leq \alpha.
\end{equation}

To control the FDR, \citet{gsell2013sequential} propose the {\em ForwardStop} rule:
\[
  \hk_{F}(Y) = \max\left\{k:\;
    -\frac{1}{k}\sum_{i=1}^k \log(1-p_i) \leq \alpha\right\}
\]
If $p_k$ is uniform, then $-\log(1-p_k)$ is an $\text{Exp}(1)$ random variable with expectation 1; thus, the sum can be seen as an estimate of the false discovery proportion (FDP):
\[
\widehat{\text{FDP}}_k = -\frac{1}{k}\sum_{i=1}^k \log(1-p_i),
\]
and we choose the {\em largest} model with \smash{$\widehat{\text{FDP}}_k \leq \alpha$}. \citet{gsell2013sequential} show that ForwardStop controls the FDR if the null $p$-values are independent of each other and of the non-nulls.

\citet{li2015accumulation} generalize ForwardStop, introducing the family of {\em accumulation tests}, which replace $-\log(1-p_i)$ with a generic {\em accumulation function} $h: [0,1] \rightarrow [0,\infty]$ satisfying \smash{$\int_{t=0}^1 h(t)dt=1$}. \citet{li2015accumulation} show that accumulation tests control a modified FDR criterion provided that the null $p$-values are $U[0,1]$, and are independent of each other and of the non-nulls.

\subsection{Ordered Testing of Random Hypotheses}\label{sec:random-hyp}

In our problem setting, the sequence of selected models is random; thus, the truth or falsehood of each $H_k$ is not fixed, as assumed in the previous section. Trivially, we can recover the guarantees from the fixed setting if we construct our $p$-values conditional on the entire path $M_{0:d}(Y)$; however, this option is rather unappealing for both computational and statistical reasons. In this section we discuss what sort of conditional control the single-step $p$-values must satisfy to recover each of the guarantees in Section~\ref{sec:orderedProposals}.

We note that conditioning on the current null model does guarantee that $p_k$ is uniform conditional on the event $\{F\in M_{k-1}(Y)\}$, the event that the $k$th null is true. Unfortunately, however, it does {\em not} guarantee that the stopping rules of Section~\ref{sec:orderedProposals} actually control FDR or FWER. Let $z_{\alpha/2}$ denote the upper-$\alpha/2$ quantile of the $\cN(0,1)$ distribution, and consider the following counterexample.

\begin{restatable}{proposition}{counterex}\label{prop:counterexample}
Consider linear regression with $n=p=3$, known $\sigma^2=1$ and identity design $X=I_3$, and suppose that we construct the model path as follows: if $|Y_3|>z_{\alpha/2}$, we add $X_1$ to the active set first, then $X_2$, then $X_3$. Otherwise, we add $X_2$, then $X_1$, then $X_3$. At each step we construct selective max-$z$ $p$-values and finally choose \smash{$\hk$} using BasicStop. 

If $\mu=(0,C,0)$, then the FWER for this procedure tends to $2\alpha-\alpha^2$ as $C\to\infty$.
\end{restatable}

A proof of Proposition~\ref{prop:counterexample} is given in the appendix.  The problem is that $k_0$ is no longer a fixed index. Consequently, even though $\cL(p_k \mid H_k \text{ true})=U[0,1]$ for each fixed $k$, the $p$-value \smash{$p_{k_0}$} corresponding to the {\em first} true null hypothesis is stochastically smaller than uniform. In this counterexample, $k_0=2 \iff |Y_3|>z_{\alpha/2}$, giving $\text{FWER}= 1$ conditional on the event $k_0=2$ and leaving no room for error when $k_0\neq 2$.

If, however, we could guarantee that for each $k=1,\ldots, d$,
\[
\P(p_k \leq \alpha \mid H_k \text{ true},\, k_0=k) \leq \alpha
\]
then BasicStop again would control FWER, by \eqref{eq:basic-stop-control}. Note that $k_0$ is unknown, so we cannot directly condition on its value when constructing $p$-values. However, because
\[
\{k_0=k\} = \{F\in M_{k-1}, F\notin M_{k-2}\},
\]
it is enough to condition on $(M_{k-2}, M_{k-1})$. As we will see in Section~\ref{sec:pValsIndep}, conditioning on $(M_{k-1}, U_{k-1})$ is equivalent to conditioning on $(M_{0:(k-1)}, U_{k-1})$ in most cases of interest including forward stepwise likelihood paths and ever-active regularization paths (but not the path in Proposition~\ref{prop:counterexample}).

Similarly, the error control guarantees of \citet{gsell2013sequential} and \citet{li2015accumulation} do not directly apply to case where the null hypotheses are random. However, we recover these guarantees if we have {\em conditional} independence and uniformity of null $p$-values given $k_0$: that is, if for all $k=0,\ldots, d-1$ and $\alpha_{k+1},\ldots,\alpha_{d} \in [0,1]$, we have
\begin{equation}\label{eq:indepCond_random_k0}
  \P_F(p_{k+1} \leq \alpha_{k+1}, \ldots, p_d \leq \alpha_d
  \mid p_1, \ldots, p_k) \eqAS \left(\prod_{i=k+1}^d \alpha_i\right) \quad \text{on } \{k_0(Y, F) = k\},
\end{equation}
For the sake of brevity, we will say that $p$-value sequences satisfying~\eqref{eq:indepCond_random_k0} are {\em independent on nulls}.

The following proposition shows that independence on nulls allows us to transfer the error-control guarantees of ForwardStop and accumulation tests to the random-hypothesis case.
\begin{proposition}
  Let $\hk$ be a stopping rule operating on $p$-values $p_{1:d}(Y)$ for nested hypotheses $H_{1:d}(Y)$. For some function $g$, let
  \[
  \Err_g = \E\left[ g\left(\hk(p_{1:d}(Y)), k_0(Y, F)\right)\right],
  \]
  Suppose that \smash{$\hk$} controls $\Err_g$ at level $\alpha$ if $H_{1:d}$ are fixed and the null $p$-values are uniform and independent of each other and the non-nulls. Then, \smash{$\hk$} controls $\Err_g$ at level $\alpha$ whenever $p_{1:d}$ satisfy~\eqref{eq:indepCond_random_k0}.
\end{proposition}
\begin{proof}
    For nested hypotheses, $k_0$ completely determines the truth or falsehood of $H_k(Y)$ for every $k$. If \smash{$\hk$} controls $\Err_g$ in the 
fixed-hypothesis case, it must in particular control $\Err_g$ conditional on any fixed values of $p_{1:k_0}$, since we could set 
\[
p_{1:k_0} \sim \prod_{k=1}^{k_0} \delta_{a_k}
\]
for any sequence $a_1,\ldots, a_{k_0}$, where $\delta_a$ is a point mass at $a \in [0,1]$.

Thus, \eqref{eq:indepCond_random_k0} implies
  \[
  \E\left[ g\left(\hk(p_1,\ldots,p_d), k_0(Y, F)\right) \;\mid\;
      k_0(Y, F), \; p_{1:k_0}(Y)\right] \leqAS \alpha.
  \]
  Marginalizing over $(p_{1:k_0},k_0)$ gives $\Err_g\leq \alpha$.
\end{proof}

Note that~\eqref{eq:indepCond_random_k0} implies in particular that each $p_k$ is uniform given $k_0=k$; thus, independence on nulls is enough to guarantee that BasicStop and ForwardStop control FWER and FDR, respectively. The next section discusses conditions on the model sequence $M_{0:d}$ and the $p$-value sequence $p_{1:d}$ under which~\eqref{eq:indepCond_random_k0} is satisfied.

\section{Conditions for Independent $p$-values}\label{sec:pValsIndep}

We now develop sufficient conditions for constructing $p$-values with the {\it independent on nulls} property (\ref{eq:indepCond_random_k0}). To begin, we motivate the general theory by discussing a specific case, the max-$t$ test for forward stepwise regression.

\subsection{Independence of max-$t$ $p$-values}

Recall that at step $k$, the max-$t$ test rejects for large \smash{$T_k=|t_{E_{k-1}}^*|$}, comparing its distribution to the conditional law
\begin{equation}\label{eq:maxTnull}
\cL\left(T_k \; \mid \; E_{k-1}, \,X_{E_{k-1}}'Y, \,\|Y\|^2 \right).
\end{equation}
This conditional distribution is the same for all $F\in M(E_{k-1})$, because \smash{$U_{k-1}=(X_{E_{k-1}}'Y,\; \|Y\|^2)$} is the complete sufficient statistic for $M(E_{k-1})$.

If $F\in M(E_{k-1})$, then $p_k$ is uniform and independent of the previous $p$-values $p_{1:(k-1)}$ since:
\begin{enumerate}
\item $p_k$ is uniform conditional on $E_{k-1}$ and $U_{k-1}$ by construction, and
\item $p_1, \ldots, p_{k-1}$ are functions of $E_{k-1}$ and $U_{k-1}$, as we will see shortly.
\end{enumerate}
Informally, the pair $(E_{k-1}, U_{k-1})$ forms a ``wall of separation'' between $p_k$ and $p_{1:(k-1)}$, guaranteeing that $\cL(p_k \mid p_{1:(k-1)}) = U[0,1]$ whenever $H_k$ is true.

Next we will see why $p_1,\ldots,p_{k-1}$ are functions of $(E_{k-1},U_{k-1})$. Observe that knowing $E_{k-1}=E$ tells us that the $k-1$ variables in $E$ are selected first, and knowing \smash{$U_{k-1}=(X_{E_{k-1}}'Y,\|Y\|^2)$} is enough information to compute the $t$-statistics $t_{j,D}$ for all $j\in E$ and $D\sub E$. As a result, we can reconstruct the order in which those $k-1$ variables were added. In other words, $(E_i, U_i)$ is a function of $(E_{k-1}, U_{k-1})$ for all $i<k$.

Furthermore, for $i<k$, $p_i$ is computed by comparing $T_i=|t_{j_i, E_{i-1}}|$, which is a function of $(E_i, U_i)$, to the reference null distribution $\cL_{H_i}(T_i \mid E_{i-1}, U_{i-1})$, which is a function of $(E_{i-1}, U_{i-1})$. Having verified that under $H_k$, $p_k$ is uniform and independent of $p_{1:{k-1}}$, we can apply this conclusion iteratively to see that all remaining $p$-values are also uniform and independent.

By contrast, the saturated model $p$-values are computed using a reference distribution different from~\eqref{eq:maxTnull}, one that depends on
information not contained in $(E_k,U_k)$. As a result, saturated-model $p$-values are generally not independent on nulls. We discuss the regression setting in more detail in Section \ref{sec:selective-reg}.

\subsection{General Case}

We now present sufficient conditions for independence on nulls generalizing the above ideas, culminating in Theorem~\ref{thm:suffCond} at the end of this section.

Define the {\em sufficient filtration} for the path $M_{0:d}$ as the filtration $\sF_{0:d}$ with 
\[
\sF_k = \sF(M_{0:k}, U_k),
\]
where $\sF(Z)$ denotes the $\sigma$-algebra generated by random variable $Z$. By our assumption of minimal sufficiency, $\sF_i \sub \sF_k$ for $i\leq k$.

For most path algorithms of interest, including forward stepwise and regularized likelihood paths, observing $(M_{k-1}, U_{k-1})$ is equivalent to observing $(M_{0:k-1},U_{k-1})$, because knowing $(M_{k-1}, U_{k-1})$ is enough to  reconstruct the subpath $M_{0:k-1}$. We say that $M_{0:d}$ satisfies the {\em subpath sufficiency principle} (henceforth SSP) if \begin{equation}\label{eq:SSP}
  \sF(M_{0:k}, U_k) = \sF(M_k, U_k), \quad k=0,\ldots, d.
\end{equation}

A valid {\em selective $p$-value} $p_k$ for testing $H_k:\; F \in M_{k-1}(Y)$ satisfies, for any $F\in M_{\infty}$,
\[
\P_F(p_k \leq \alpha \mid M_{k-1}) \leqAS \alpha \quad \text{ on } \{F\in M_{k-1}\}.
\]
We say that a filtration $\sF_{0:d}$ {\em separates} the $p$-values $p_{1:d}$ if
\begin{enumerate}
\item $p_k(Y)$ is super-uniform given $\sF_{k-1}$ on the event $\{F\in M_{k-1}(Y)\}$, and
\item $M_k(Y)$ and $p_k(Y)$ are measurable with respect to $\sF_k$.
\end{enumerate}

If we think of $\sF_k$ as representing information available at step $k$, then the first condition means that $p_k$ ``excludes'' whatever evidence may have accrued against the null by step $k-1$, and the second means that any information revealed after step $k$ is likewise irrelevant to determining $p_k$. Separated $p$-values are independent on nulls, as we see next.

\begin{proposition}[Independence of 
  Separated $p$-Values]\label{prop:jointConserv}

  Let $p_{1:d}$ be selective $p$-values for $M_{0:d}$, 
  separated by $\sF_{0:d}$.

  If the $p$-values are exact then $p_{k+1}, \ldots, p_d$ are
  independent and uniform given $\sF_k$ on the event $\{k_0=k\}$.
  If they are super-uniform, then for all
  $\alpha_{k+1},\ldots,\alpha_d \in [0,1]$,
  \begin{equation}\label{eq:jointSuper}
  \P_F\left(p_{k+1}\leq \alpha_{k+1}, \ldots, p_d \leq \alpha_d, \;
    \mid\; \sF_k\right) \leqAS \left(\prod_{i=k+1}^d
  \alpha_i\right) \quad \text{ on } \{k_0 = k\}.
  \end{equation}
\end{proposition}

\begin{proof}
  Noting that $\{k_0=k\}$ is $\F_k$-measurable, it is enough to show that 
  \begin{equation}\label{eq:jointSuper2}
  \P_F\left(p_{k+1}\leq \alpha_{k+1}, \ldots, p_d \leq \alpha_d, \;
    \mid\; \sF_k\right)\1_{\{k_0\leq k\}} \leqAS \left(\prod_{i=k+1}^d
  \alpha_i\right)\1_{\{k_0\leq k\}}
  \end{equation}
  We now prove~\eqref{eq:jointSuper2} by induction. Define $B_i = \{p_i \leq \alpha_i\}$. The base case is
  \[
  \P_F\left(B_d \mid \sF_{d-1}\right)1_{\{k_0 \leq d-1\}} \leqAS \alpha_d
1_{\{k_0 \leq d-1\}},
  \]
  which is true by construction of $p_d$. 
  For the inductive case, note that 
  \begin{align*}
    \P_F\left(B_{k+1}, \ldots, B_d
      \mid \sF_k\right)1_{\{k_0 \leq k\}} 
    &\eqAS \E_F\bigg[ 1_{B_{k+1}} 
    \P_F\left(B_{k+2}, \ldots, B_d
      \mid \sF_{k+1}\right)1_{\{k_0 \leq k+1\}}
    \mid \sF_k\bigg]1_{\{k_0 \leq k\}}\\
    &\leqAS \P_F\left[ B_{k+1}
      \mid \sF_k\right]1_{\{k_0 \leq k\}}\prod_{i=k+2}^d \alpha_i \\
    &\leqAS \left(\prod_{i=k+1}^d \alpha_i \right)1_{\{k_0 \leq k\}} .
  \end{align*}
Lastly, if the $p_k$ are exact then the inequalities above become equalities, implying uniformity and mutual independence.
\end{proof}

The sufficient filtration separates $p_{1:d}$ if and only if (1) $p_k$ is super-uniform given $M_{0:(k-1)}$ and $U_{k-1}$, and (2) $p_k$ is a function of $M_{0:k}$ and $U_k$.

To be a valid selective $p$-value for a single step, $p_k$ only needs to be super-uniform conditional on $M_{k-1}$. It may appear more stringent to additionally require that $p_k$ must condition on the entire subpath $M_{0:(k-1)}$ as well as $U_{k-1}$, but in practice there is often no difference: if $M_{0:d}$ satisfies the SSP and each $U_{k-1}$ is a complete sufficient statistic, then {\em every} exact selective $p$-value $p_k$ is also uniform conditional on $\sF_{k-1}$.

The requirement that $p_k$ must be $\sF_k$-measurable has more bite. For example, we will see in Section~\ref{sec:selective-reg} that it excludes saturated-model tests in linear regression.

Collecting together the threads of this section, we arrive at our main result: a sufficient condition for $p_{1:d}$ to be independent on nulls.
\begin{theorem}[Sufficient Condition for Independence on Nulls] \label{thm:suffCond}
Assume that $U(Y; M)$ is a complete sufficient statistic for each candidate model $M$, that each $p_k$ is an exact selective $p$-value for $M_{k-1}$, and that the path algorithm $M_{0:d}(\cdot)$ satisfies the SSP.

If $p_k$ is $\sF_k$-measurable for each $k$, then the $p$-value sequence $p_{1:d}(Y)$ is independent on nulls.
\end{theorem}

\begin{proof}
We apply the definition of completeness to the function 
\[
h(U(Y; \, M)) \;=\; 
\alpha - \P\bigg(p_k(Y) \leq \alpha \mid U(Y; \, M), \, M_{k-1}(Y) = M\bigg)
\]
If $p_k$ is exact given $M_{k-1} = M$, then 
\[
\E_F[h(U_{k-1}) \mid M_{k-1}=M] \;=\; 0
\] 
for every $F\in M$. Therefore, we must have $h(U_{k-1})\eqAS 0$, so $p_k$ is independent of $\sF(M_k, U_{k-1})$. Because $M_{0:d}$ satisfies the SSP, $p_k$ is also independent of $\sF_{k-1} = \sF(M_{0:(k-1)}, U_{k-1})$.

If $p_k$ is also $\sF_k$-measurable, then the sufficient filtration separates $p_{1:d}$, implying that the sequence $p_{1:d}$ is independent on nulls.
\end{proof}

\paragraph{Remark} If our path algorithm does not satisfy the SSP, we can repair the situation by constructing $p$-values that are uniform conditional on $(M_{0:k-1},U_{k-1})$.

\subsection{Independence for Greedy Likelihood and Next-Entry $p$-values}\label{sec:indep-greedy-entry}

In this section, we apply Theorem~\ref{thm:suffCond} to establish that in the generic sparse parametric model of Section~\ref{sec:genSparse}, the forward stepwise path and all ever-active regularization paths satisfy the SSP. Moreover, the greedy likelihood ratio statistic $G_k$ and the next-entry statistic $\lambda_k$ are $\sF_k$-measurable with respect to the sufficient filtrations of the forward stepwise and ever-active regularization paths, respectively. As a result, the $p$-value sequences in each setting are independent on nulls, per Theorem~\ref{thm:suffCond}.

We begin by proving both paths satisfy the SSP:
\begin{proposition}\label{prop:forwardSSP}
Forward stepwise paths and ever-active regularization paths both satisfy the subpath sufficiency principle.
\end{proposition}

\begin{proof}
First define the restricted log-likelihood
\begin{equation*}
\ell_E(\theta; Y) = \left\{\begin{matrix} 
    \ell(\theta; \;Y) & \theta \in \Theta(E)\\ 
    -\infty  & \mathrm{ otherwise.}\end{matrix}\right. 
\end{equation*}
For some fixed step $k$ and active set $E$, denote the event $A=\{M_k = M(E)\}$. Conditioning on $U$, the restricted likelihood is proportional to a function depending only on $U = U(Y; \;M(E))$. The log-likelihood decomposes as
\[
\ell_E(\theta; Y) = \ell_E^{U}(\theta; U) 
+ \ell_E^{Y \mid U}(Y \mid U).
\]
The second term, the log-likelihood of $Y$ given $U$, does not depend on $\theta$ because $U$ is sufficient.

Recall the forward stepwise path is defined by
  \begin{equation}\label{eq:forwardStepDef}
    j_k = \argmax_j \;\sup \left\{\ell(\theta; Y):\; \theta\in\Theta(E_{k-1} \cup \{j\})\right\}, \quad \text{ and } \quad
    E_k = E_{k-1} \cup \{j_k\}.
  \end{equation}

  On $A$, we have $E_s(Y) \sub E$ for all $s\leq k$, meaning that the maximum in~\eqref{eq:forwardStepDef} is attained by some $j\in E$ at every step. So, for $s \leq k$, we have
\begin{align*}
  j_s &= \argmax_{j\in E} \;\sup \left\{\ell_E(\theta; Y):\;
    \theta\in\Theta(E_{s-1} \cup \{j\})\right\} \\
  &= \argmax_{j\in E} \;\sup \left\{\ell_E^U(\theta; U(Y)):\;
    \theta\in\Theta(E_{s-1} \cup \{j\})\right\}. 
\end{align*}
The above shows that $j_1,\ldots, j_{k-1}$ all depend on $Y$ only through $U(Y)$, which equals $U_k(Y)$ on $A$. As a result, it also follows that the entire sequence $M_{0:k}(Y)$ depends only on $(M_k, U_k)$.

As for ever-active regularization paths, if we denote 
\[
\hat\theta^{(E,\lambda)} = \argmin_{\theta\in\Theta(E)} \; -\ell_E^{U}(\theta; U(Y)) + P_\lambda(\theta),
\]
then 
\[
\hat\theta^{(E,\lambda)} \eqAS \hat\theta^{\lambda} \text{ on } 
A \cap 1_{\{\lambda \geq \lambda_k\}}.
\]
But \smash{$\hat\theta^{(E,\lambda)}$} depends only on $U(Y)$. Therefore, on $A$, we can reconstruct the entire path of solutions for $\{\lambda\in \Lambda:\; \lambda\geq \lambda_k\}$, once we know $(M_k,\,U_k)$.
\end{proof}

As a direct consequence of Proposition~\ref{prop:forwardSSP}, greedy likelihood $p$-values for the forward-stepwise path, and next-entry $p$-values for ever-active regularization paths, are independent on nulls.

\begin{corollary}\label{cor:greedyLikIndep}
  If $U(Y; M)$ is complete and sufficient for model $M$, then the selective greedy likelihood $p$-values computed for the forward stepwise path are independent on nulls.
\end{corollary}
\begin{proof}
  Per Theorem~\ref{thm:suffCond}, it is enough to show that $p_k$ is $\sF_k$-measurable with respect to the sufficient filtration for the forward stepwise path. First, the test statistic $G_k$ is $\sF_k$-measurable because
  \begin{align*}
    G_k(Y) &= \sup_{\theta\in \Theta(E_k)} {\ell(\theta; Y)} \;\;- \sup_{\theta\in \Theta(E_{k-1})} {\ell(\theta; Y)}\\
    &= \sup_{\theta\in \Theta(E_k)} {\ell_{E_k}^{Y\mid U}(\theta; U_k)} \;\;- \sup_{\theta\in \Theta(E_{k-1})} {\ell_{E_k}^{Y\mid U}(\theta; U_k)}.
  \end{align*}
  Second, the null distribution $\cL\left(G_k \,\mid\, M_{k-1}, U_{k-1}\right)$ against which we compare $G_k$ is also $\sF_k$-measurable because it depends only on $(M_{k-1}, U_{k-1})$, both of which are $\sF_{k-1}$-measurable.
\end{proof}

\begin{corollary}\label{cor:nextLambdaIndep}
  If $U(Y; M)$ is complete and sufficient for model $M$, then the selective next-entry $p$-values computed for any ever-active regularized likelihood path are independent on nulls.
\end{corollary}
\begin{proof}
  Per Theorem~\ref{thm:suffCond}, it is enough to show that $p_k$ is $\sF_k$-measurable with respect to the  sufficient filtration for the regularized-likelihood path.  Recall that 
  \[
  \Lambda_k = \{\lambda\in \Lambda:\; \tE_\lambda \supsetneq \tE_{\lambda_{k-1}}\},
  \quad 
  \lambda_k = \sup\Lambda_k,
  \quad \text{ and } \quad
  E_k = \bigcap_{\lambda\in\Lambda_k} \tE_{\lambda},
  \]
  and that the \smash{$\tE_{\lambda}$} are nested by construction and finite by assumption. As a result, there is some $\gamma_k\in \Lambda_k$ for which \smash{$\tE_{\gamma_k} = E_k$}. 
  
  Furthermore, for all $\lambda\in\Lambda, \lambda\geq \gamma_k$, we have \smash{$\tE_\lambda \sub E_k$}. As a result, \smash{${\hat\theta^{\lambda} =\hat\theta^{(E_k,\lambda)}}$} for such $\lambda$, so $\tE_{\lambda}$ depends only on \smash{$\ell_{E_k}^{Y\mid U}(Y;\, U_k)$}, and therefore $\lambda_k$ can be computed from $(E_k, U_k)$.  Second, the null distribution $\cL\left(\lambda_k \,\mid\, M_{k-1}, U_{k-1}\right)$ against which we compare $\lambda_k$ is also $\sF_k$-measurable because it depends only on $(M_{k-1}, U_{k-1})$, both of which are $\sF_{k-1}$-measurable.
\end{proof}

As we have shown, the selective max-$t$ test is equivalent to the selective greedy likelihood test in the case of linear regression. In the next section, we will compare and contrast the max-$t$ and next-entry, and other {\em selected-model tests}, with the {\em saturated-model tests} proposed by \citet{tibshirani2014exact} and others. 

\section{Selective $p$-Values in Regression}
\label{sec:selective-reg}
In linear regression, \citet{fithian2014optimal} draw a distinction between two main types of selective test that we might perform at step $k$: tests in the selected model, and tests in the saturated model. For simplicity, we will assume throughout this section that our path algorithm adds exactly one variable at each step. We also assume the path algorithm satisfies the SSP, so that we need not worry about the distinction between conditioning on $(M_{k-1}, U_{k-1})$ and conditioning on $(M_{0:(k-1)}, U_{k-1})$.

\subsection{Saturated-Model Tests}
Many conditional selective tests for linear regression use the saturated-model framework. See for example \citet{taylor2013tests}, \citet{tibshirani2014exact}, \citet{lee2013exact}, and \citet{loftus2014significance}. Define the saturated model as
\[
M_{\text{sat}}:\; Y \sim \cN(\mu, \sigma^2I_n),
\]
so named because there is a mean parameter for every observation. Because the parameter $\mu$ has the same dimension as the data $Y$, we must assume that $\sigma^2$ is known. If $\sigma^2$ is unknown, we can estimate $\sigma^2$ and apply the test while plugging in the estimated value; this is what we have done with the diabetes data to compute the $p$-values in Table~\ref{tab:diab}.  \citet{tibshirani2015uniform} discuss this and other strategies for unknown $\sigma^2$.

Saturated-model tests perform inference on selected linear functionals of $\mu$. In the context of sequential linear model selection,  we can use the saturated-model framework to test the incremental null \smash{$H^{\text{inc}}:\; \theta_{j_k}^{E_k} = 0$}, where $\theta^E$ minimizes \smash{$\|\mu - X_E\theta\|_2$}.

Let $\eta_k\in \R^n$ denote the vector for which \smash{$\theta_{j_k}^{E_k} = \eta_k'\mu$}, and let \smash{$\proj_{\eta_k}^\perp$} denote the projection operator into the space orthogonal to $\eta_k$. The UMPU saturated-model test for $\eta_k'\mu$ is based on the distribution \smash{$\cL\left( \eta_k'Y \mid E_{k-1}, \;j_k, \;\proj_{\eta_k}^\perp Y \right)$}, or equivalently
\begin{equation}\label{eq:satModel}
\cL\left( X_k'Y \mid E_{k-1}, \;j_k, \;\proj_{\eta_k}^\perp Y \right)
\end{equation}

While the test statistic $X_{j_k}'Y$ is measurable with respect to $\sF_k$, the null distribution we compare it against is not: it depends on \smash{$\proj_{\eta_k}^\perp Y$}, which is not measurable with respect to $\sF_k$. As a result, saturated-model $p$-values are in general {\em not} independent on nulls.

\subsection{Selected-Model Tests}
Selected-model inference, introduced in \citet{fithian2014optimal},
represents a more powerful option when testing for inclusion of variable $j_k$. While saturated-model inference always uses the same model $M_{\text{sat}}$, selected-model inference uses the lower-dimensional statistical model chosen by our selection procedure, similarly to data splitting. In the sequential case this means that at step $k$, we test the selected null model $M(E_{k-1})$, either against the next selected model $M(E_k)\setminus M(E_{k-1})$ or against the upper model $M_\infty\setminus M(E_{k-1})$.

\subsubsection{Testing Against the Next Model}\label{sec:identify}
One option for selected-model inference is to perform a selective test of $M_{k-1}$ against $M_{k}\setminus M_{k-1}$, conditioning on both models. Conditioning on $(M_{k-1},M_k)$ is equivalent to conditioning on $(E_{k-1}, j_k)$. If $\sigma^2$ is known, the test is based on
\begin{equation}\label{eq:selModel_cond}
\cL\left(X_{j_k}'Y \mid E_{k-1}, \; j_k, \; X_{E_{k-1}}'Y\right).
\end{equation}
If $\sigma^2$ were unknown we would also need to condition on $\|Y\|^2$. We can construct $p$-values from the UMPU, equal-tailed, or one-sided test based on the one-parameter exponential family~\eqref{eq:selModel_cond}; the only difference is in how much of the overall level $\alpha$ is apportioned to the right or left tail \citep{fithian2014optimal}. If we use the construction in~\eqref{eq:selModel_cond} after a subpath-sufficient path algorithm, the resulting $p$-values are always guaranteed to have independence on nulls. This is because the test statistic $X_{j_k}'Y$ is measurable with respect to $\sF_k$, as is the reference null distribution.

Note the contrast between~\eqref{eq:selModel_cond} and~\eqref{eq:satModel}. In~\eqref{eq:selModel_cond}, we only need to condition on an $|E_{k-1}|$-dimensional projection of $Y$, whereas in~\eqref{eq:satModel} we condition on an $(n-1)$-dimensional projection of $Y$. Conditioning on more information can sap the power of a test; as we will see in Sections~\ref{sec:bivariate} and~\ref{sec:sparseReg}, the saturated-model tests often pay a heavy price for this extra conditioning.

\subsubsection{Testing Against the Upper Model}

To avoid conditioning on the next model $M_k$, we could instead test $M_{k-1}$ against the alternative $M_{\infty}\setminus M_{k-1}$. In that case, we could modify~\eqref{eq:selModel_cond} and base the test on
\begin{equation}\label{eq:selModel_marg}
\cL\left(X_{j_k}'Y \mid E_{k-1}, \; X_{E_{k-1}}'Y\right).
\end{equation}
Again, if $\sigma^2$ were unknown we would also need to condition on $\|Y\|^2$. We could also replace $X_{j_k}'Y$ with any other test statistic $T(Y)$. As long as $T(Y)$ is measurable with respect to $(M_k, U_k)$, and the path algorithm is subpath-sufficient, the $p$-values are guaranteed to be independent on nulls. The next-entry test for the lasso and the max-$t$ test in forward-stepwise regression with unknown $\sigma^2$ are both examples of this approach. Note that under resampling from the law in~\eqref{eq:selModel_marg}, the selected variable $j_k$ is random; by contrast, \eqref{eq:selModel_cond} conditions on $j_k$.

We will not pursue tests like those in Section~\ref{sec:identify}. It is not necessary to condition on $M_k$ to obtain the FDR and FWER guarantees that we want to establish, and so we follow the general principle of conditioning on as little information as possible. In Section~\ref{sec:sparseReg}, we empirically compare the max-$z$ test with the test that uses the same test statistic but also conditions on the identity of the next variable $j_k$. In our simulation this test, which we call the {\em max-$z$-identify test}, performs similarly to the max-$z$ test.

\subsection{Comparison of Selected- and Saturated-Model Inference}\label{sec:bivariate}

We now briefly compare the computational, conceptual, and statistical advantages and disadvantages of selected- and saturated-model tests, focusing on the problem of sequential model selection. For a detailed general discussion, see Section 5 of \citet{fithian2014optimal}.

Selected-model inference is typically more computationally involved than saturated-model inference. Because the saturated-model test conditions on $n-1$ dimensions, the resulting distribution in~\eqref{eq:satModel} is nothing more than a truncated univariate Gaussian random variable, and inference can be carried out analytically. By contrast, selected-model tests typically require MCMC sampling from a truncated multivariate Gaussian distribution of dimension $p-|E_{k-1}|$.

For some applications, a major conceptual advantage of the saturated-model approach is that, even when $M(E)$ is misspecified, $\theta^{E}$ is still a well-defined quantity on which we can perform exact inference. For example, a 95\% confidence interval for $\theta_{j_k}^{E_k}$ could cover that coefficient with probability 95\% even if $M(E_k)$ is misspecified. In the setting of sequential model selection, this advantage is less important: at an intermediate step $k$, we are not constructing intervals for $\theta_{j_k}^{E_k}$ but rather deciding whether to reject $M_{k-1}$ and move on. However, the saturated-model framework could be quite helpful after selection is finished, if we want to construct intervals for the least-squares parameters for the selected active set $E_{\hk}$.

There are several major statistical benefits to selected-model inference. First, selected-model inference allows us to drop the assumption that $\sigma^2$ is known. This is not possible for the saturated model because conditioning on both $\|Y\|^2$ and \smash{$\proj_{\eta_k}^\perp Y$} results in a degenerate conditional distribution for $\eta_k'Y$. Second, we have seen in Section~\ref{sec:pValsIndep} that tests of the form~\eqref{eq:selModel_cond} or~\eqref{eq:selModel_marg} yield independent $p$-values under general conditions, allowing us to apply  the sequential stopping rules of~\citet{gsell2013sequential} and \citet{li2015accumulation}. Finally, and most importantly, selected-model tests can be dramatically more powerful than saturated-model tests, as we illustrate in the next example.

\begin{example}[Bivariate Regression with Identity Design]\label{ex:bivariate}
  Consider forward stepwise selection in a regression model with $n=p=2$, with known $\sigma^2=1$ and identity design 
\[
X = I_2=\begin{pmatrix} 1 & 0 \\ 0 & 1\end{pmatrix}.
\] 
The forward stepwise path and the lasso path both select $j_1=1$ if and only if $|Y_1|>|Y_2|$. The selection event $A_1=\{|Y_1| > |Y_2|\}$ is shown in yellow in Figure~\ref{fig:bv_condSets}. On $A_1$, the first two models are
\[
M_0:\; Y\sim \cN(0,I_2), \quad \text{ and } \quad
M_1:\; Y \sim \cN\left(\binom{\mu_1}{0}, \; I_2\right).
\]
The selected-model test at step 1 compares $Y_1$ to its distribution under $M_0$, the global null, conditional on $A_1$.\footnote{In this very simple example, the max-$z$, max-$z$-identify, and next-entry tests are all equivalent.} By contrast, the saturated-model test is a test of $H_0:\; \mu_1=0$ in the model $M_{\text{sat}}:\; Y \sim \cN(\mu, I_2)$. Thus, it must condition on $Y_2$ to eliminate the nuisance parameter $\mu_2$, comparing $Y_1$ to its null distribution given $A_1$ and the observed value of $Y_2$.

Figure~\ref{fig:bv_condSets} shows the conditioning sets for each model when $Y=(2.9, 2.5)$. Beside it, Figure~\ref{fig:bv_nullDists} shows the null conditional distribution of the test statistic $Y_1$ for each test. The $p$-values for the selected and saturated models are 0.007 and 0.3, respectively. Figure~\ref{fig:bv} is reproduced from \citet{fithian2014optimal}, where the same example was presented in less detail.
\end{example}

\begin{figure}
  \centering
  \begin{subfigure}[t]{.4\textwidth}
    \includegraphics[width=\textwidth]{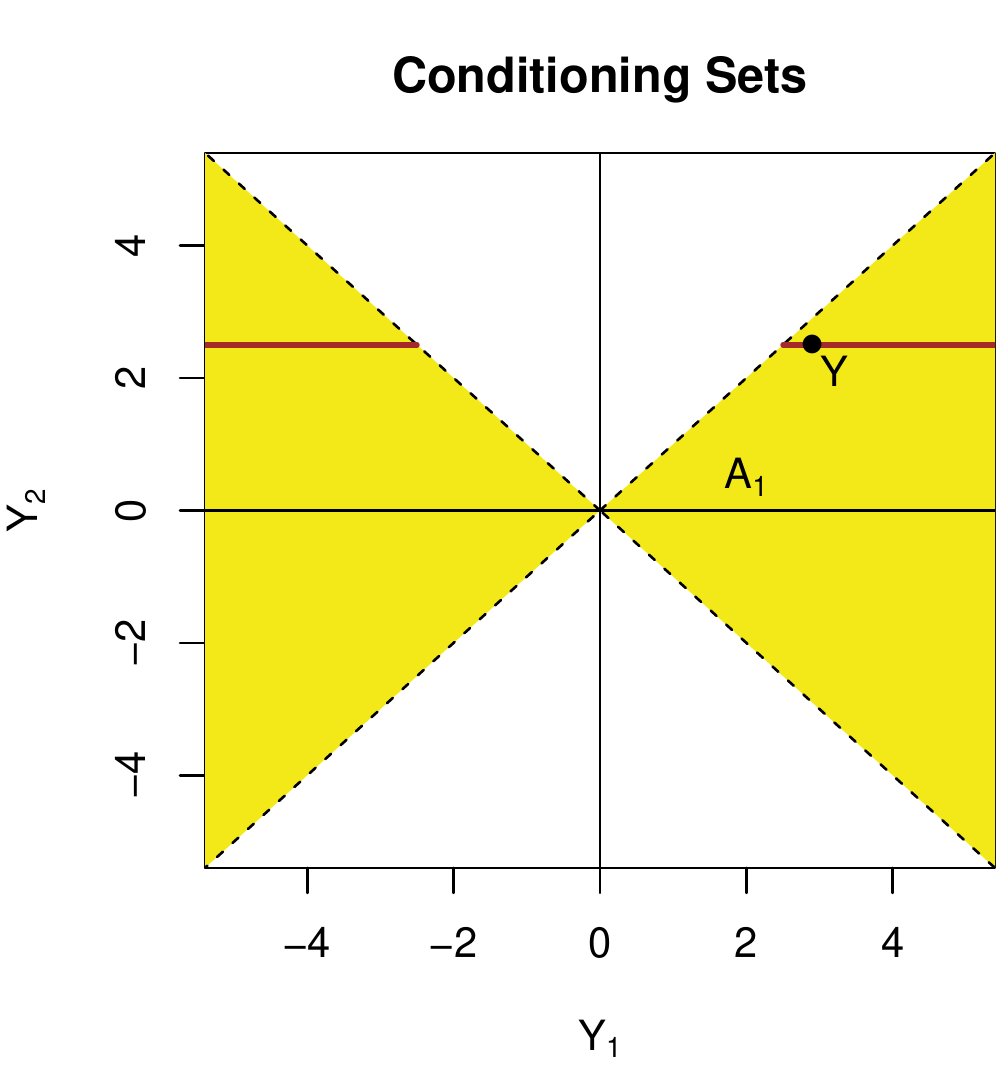}
    \caption{The selected-model test conditions on $j_1=1$ (yellow region), while the saturated-model test also conditions on $Y_2=2.5$ to eliminate the nuisance variable $\mu_2$ (brown line segments).}
    \label{fig:bv_condSets}
  \end{subfigure}
  \hspace{.1\textwidth}
  \begin{subfigure}[t]{.4\textwidth}
    \includegraphics[width=\textwidth]{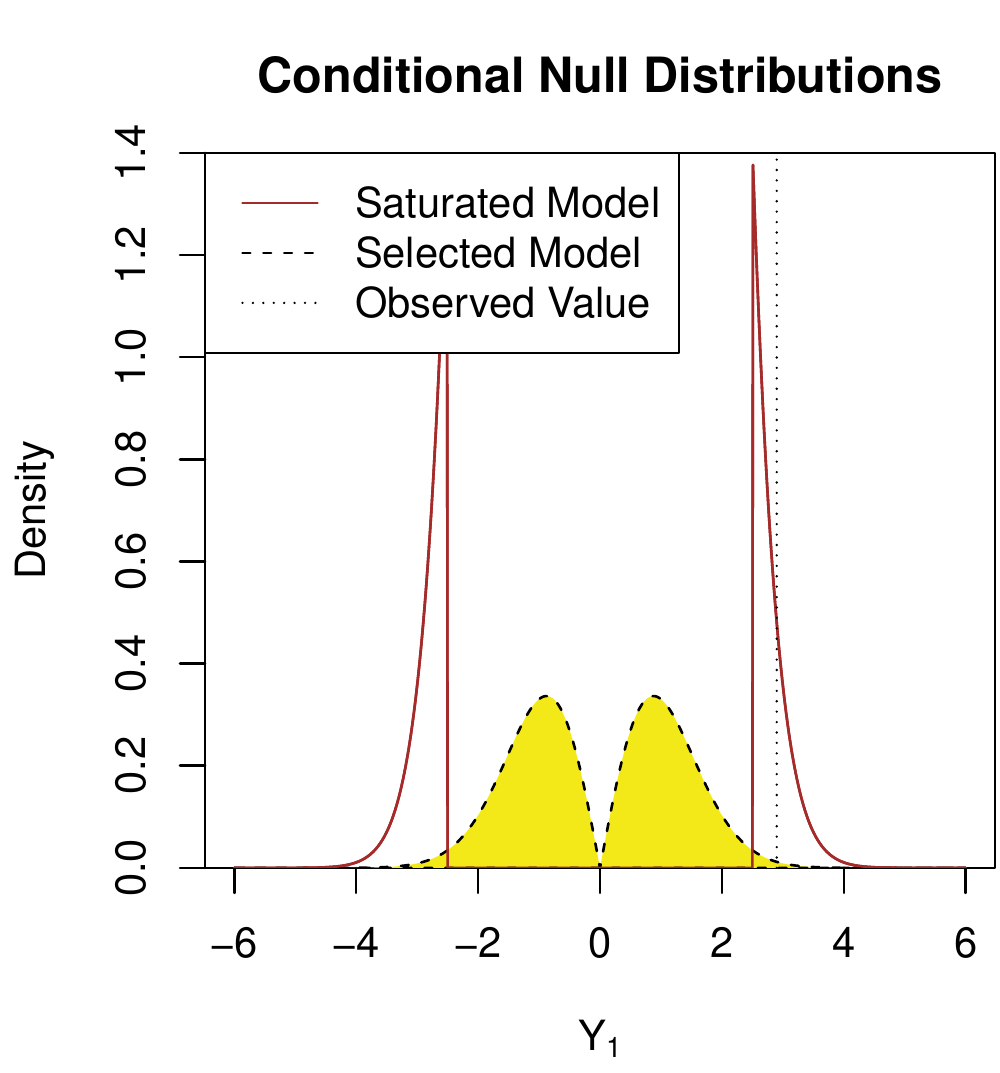}
    \caption{Conditional distributions of $Y_1$ under
      $M_0:\; Y \sim \cN(0,I_2)$. The realized value $|Y_1|=2.9$ is
      quite large given that $j_1=1$. By
      contrast, $|Y_1|=2.9$ is not especially large once we 
      also condition on $Y_2=2.5$.}
  \label{fig:bv_nullDists}
  \end{subfigure}
  \caption{Contrast between the saturated-model and selected-model
    tests in Example~\ref{ex:bivariate}. The selected-model test is based on  $\cL(Y_1 \gv j_1=1)$,  whereas the saturated-model test is based on $\cL(Y_1  \gv Y_2, \; j_1=1)$. 
    When $Y=(2.9, 2.5)$, the selected- and saturated-model $p$-values are 0.007 and 0.3, respectively.}
  \label{fig:bv}
\end{figure}

Figure~\ref{fig:bv} illustrates a general phenomenon in saturated-model tests: when there are near-ties between strong variables that are competing to enter the model, the resulting $p$-value may be very weak. Figure~\ref{fig:bv_rocCurve} displays the cumulative distribution function for the first $p$-value when $\mu=(4,4)$, a very strong signal. While the selected-model test has nearly perfect power, it is not uncommon for the saturated-model test to produce large $p$-values, even in the range of 0.5-0.9. These large $p$-values arise when there is a near tie between the variables.

\begin{figure}
  \centering
  \includegraphics[width=.5\textwidth]{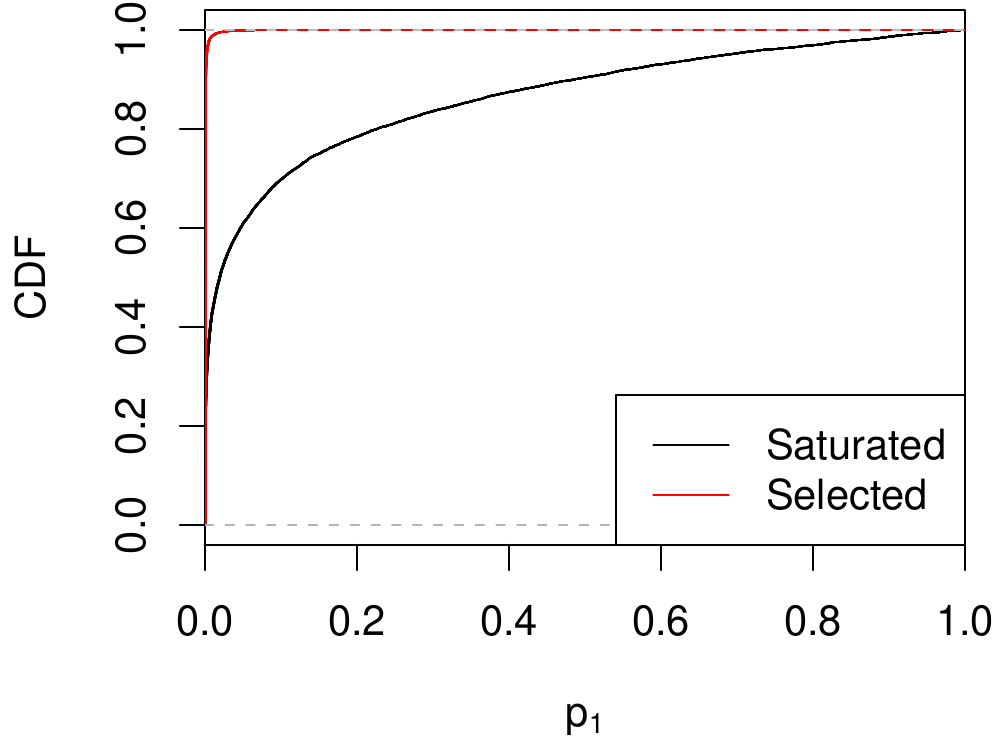}
  \caption{Cumulative distribution function of $p_1(Y)$ for the selected- and saturated-model tests when $\mu=(4,4)$. Even though the signal is very strong, the saturated-model test results in large $p$-values in realizations like the one in Figure~\ref{fig:bv_nullDists}, where there is a near-tie between $|Y_1|$ and $|Y_2|$. By contrast, the selected-model test has nearly perfect power.}
  \label{fig:bv_rocCurve}
\end{figure}

Results in \citet{fithian2014optimal} show that the selected-model test is strictly more powerful when the selected model is correct, i.e., when $\mu_2=0$. Figure~\ref{fig:bv_powCurves} shows the power curve for each test when $\mu_2=0$ (left panel) and $\mu_2=4$ (right panel). While the selected-model test is more powerful when the selected model is correct, the difference between the two is relatively small. The difference is much more pronounced when $\mu_2=4$. Note that if $\mu=(0,4)$, the incremental null is true, but the complete null is false. That is, the global null model $M_0$ is missing an important signal variable, but the missing variable is not $X_1 = (1,0)$. Because the saturated-model test is a valid test of the incremental null, its power is $\alpha=0.05$. By contrast, the selected-model test rejects about half of the time when $\mu=(0,4)$, successfully detecting the poor fit of $M_0$ to the data.

\begin{figure}
  \centering
  \begin{subfigure}[t]{.4\textwidth}
    \includegraphics[width=\textwidth]{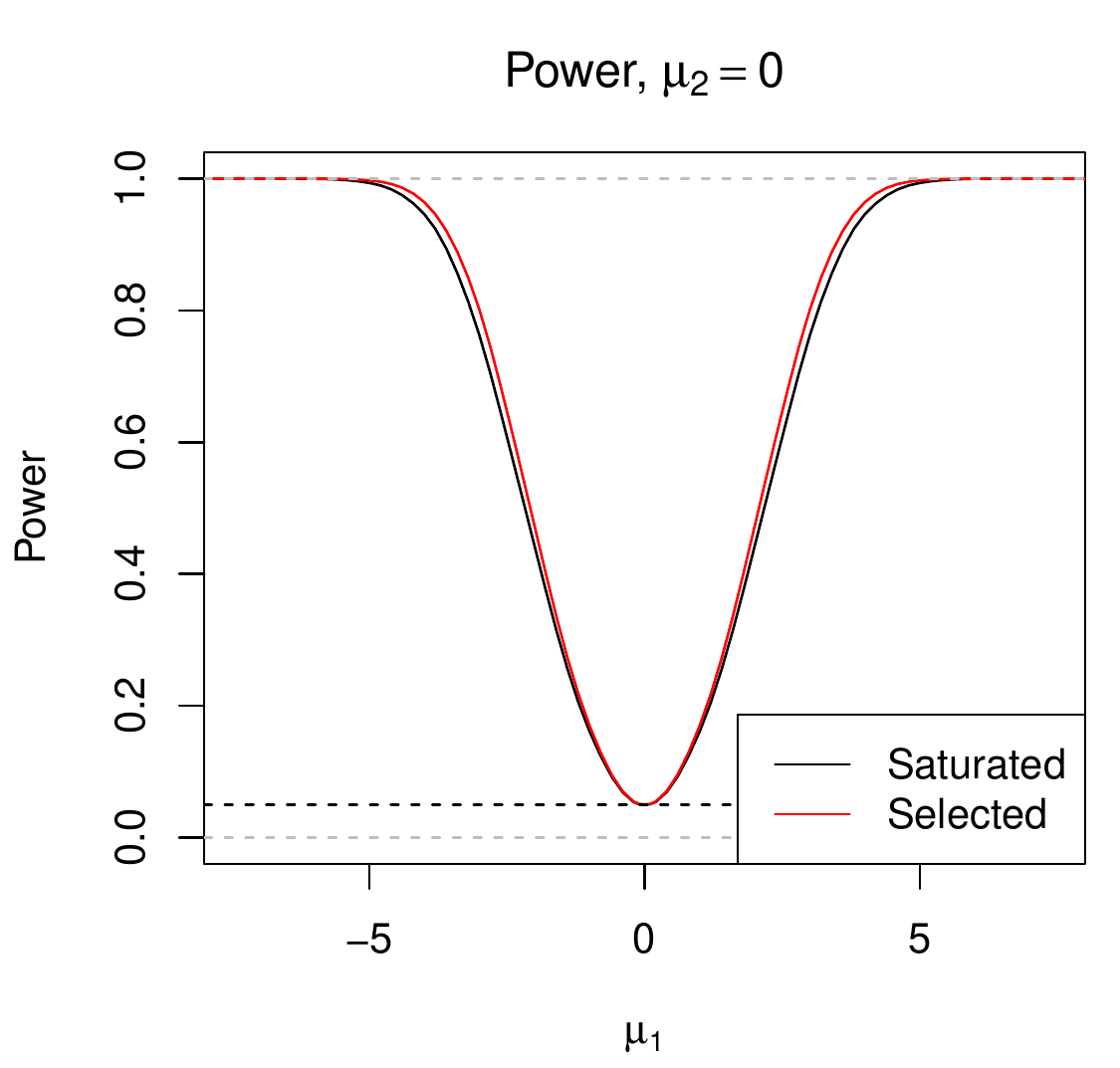}
  \end{subfigure}
  \hspace{.1\textwidth}
  \begin{subfigure}[t]{.4\textwidth}
    \includegraphics[width=\textwidth]{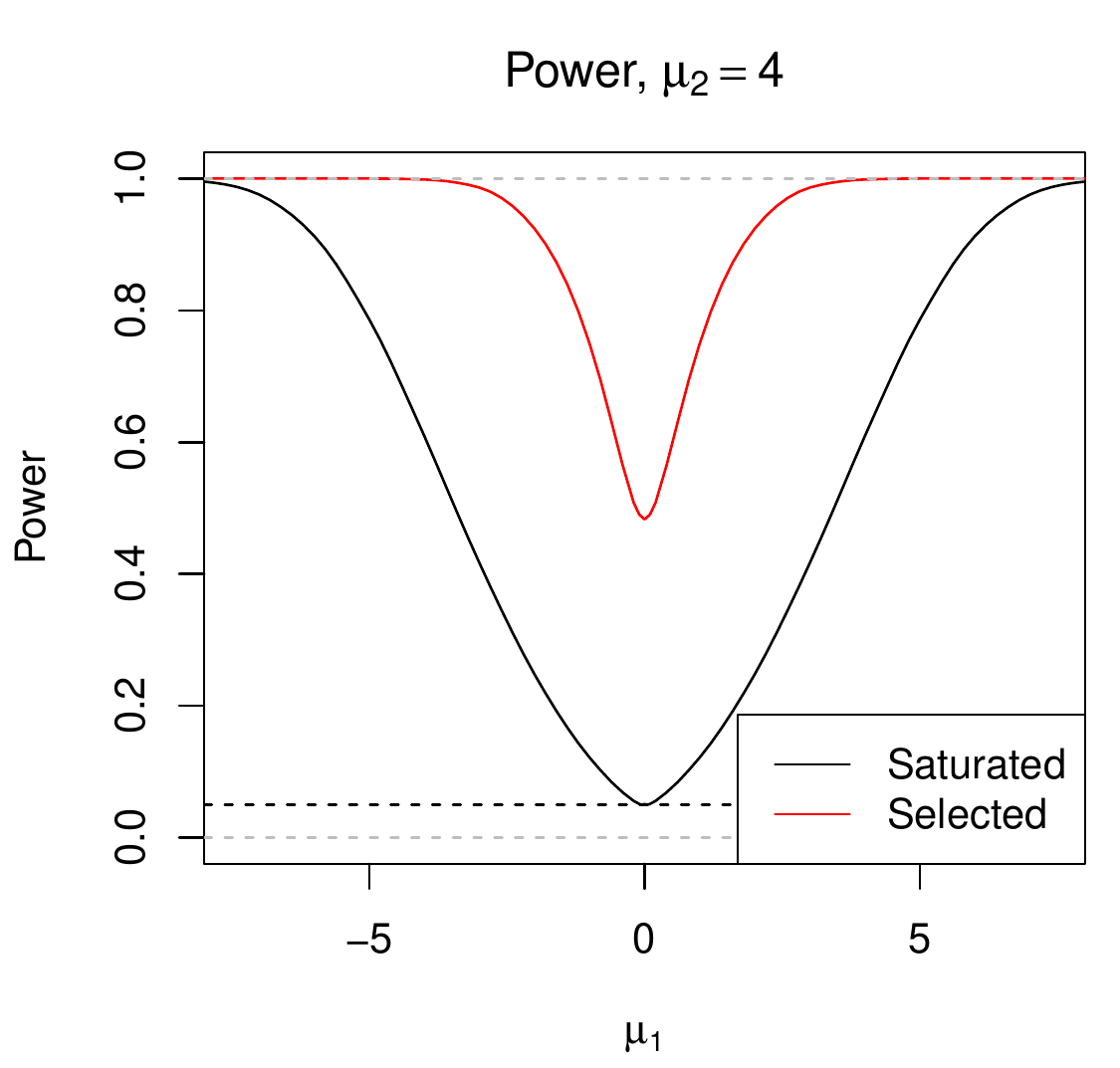}
  \end{subfigure}
  \caption{Power at level $\alpha=0.05$ for saturated- and selected-model tests at step 1, given that $j_1=1$. The power is plotted as a function of $\mu_1$, for two different values of $\mu_2$. When $\mu_2=0$, the selected-model test is strictly more powerful than the saturated-model test, but the difference is slight. By contrast, when $\mu_2=4$, the selected-model test is much more powerful. The dashed line shows $\alpha=0.05$.}
   \label{fig:bv_powCurves}
\end{figure}

Saturated-model $p$-values are generically non-independent. Continuing with Example~\ref{ex:bivariate}, Table~\ref{tab:bv_twoWayTable} shows a two-way contingency table for the saturated-model $p$-values $(p_1(Y), p_2(Y))$, binned into cells of height and width 0.2, simulated under the global null $\mu=0$. Because $k_0=0$, both $p$-values are uniform by construction, but the saturated-model $p$-values are strongly dependent, with correlation $-0.48$. By contrast, the selected-model $p$-values $(p_1,p_2)$ are independent under the global null.

\begin{table}[ht]
  \centering
  \begin{tabular}{l|ccccc|c}
    \multicolumn{7}{c}{Saturated-Model $p$-Values 
      (\% of $10^6$ Simulations)}\\[7pt]
    \hline
    \multicolumn{7}{c}{}\\[-1.5ex]
    \multicolumn{7}{c}{$p_2(Y)$}\\[5pt]
    ${\large p_1(Y)}$ & (0,0.2] & (0.2,0.4] & (0.4,0.6] & (0.6,0.8] & (0.8,1] & \textbf{Total} \\ 
    \hline
    (0,0.2] & 1.0 & 2.7 & 4.2 & 5.6 & 6.7 & 20.1 \\ 
    (0.2,0.4] & 1.4 & 3.4 & 4.5 & 5.2 & 5.5 & 20.0 \\ 
    (0.4,0.6] & 2.3 & 4.3 & 4.7 & 4.5 & 4.2 & 20.0 \\ 
    (0.6,0.8] & 4.2 & 5.4 & 4.3 & 3.4 & 2.7 & 20.0 \\ 
    (0.8,1] & 11.1 & 4.3 & 2.3 & 1.4 & 1.0 & 20.0 \\ 
    \hline
    \textbf{Total} & 19.9 & 20.0 & 20.0 & 20.1 & 20.0 & 100.0 \\ 
    \hline
  \end{tabular}
  \caption{Two-way contingency table of saturated-model $p$-values $(p_1(Y), p_2(Y))$ for Example~\ref{ex:bivariate}, after binning into cells of height and width 0.2. We report the percentage of $p$-value pairs falling into each cell out of one million simulations from the global null hypothesis, $\mu=0$. Both $p$-values are marginally uniform but strongly dependent, with a correlation of $-0.48$.}
\label{tab:bv_twoWayTable}
\end{table}

\section{Computation}
\label{sec:computation}

Until this point, we have deferred discussing how to compute the reference distribution for the max-$t$ test, next-entry test, or any of the other tests we have discussed. In each case, the $p$-value at step $k$ comes from a test that rejects when some test statistic $T_k$ is large compared to the law
\begin{equation}\label{eq:genDist}
\cL(T_k(Y) \,\mid\, M_{k-1}(Y), U_{k-1}(Y)),
\end{equation}
under the null hypothesis that $F\in M_{k-1}$. Because $U_{k-1}$ is sufficient for $Y$ under the null, this law is completely known. Thus, we have reduced computation of the $p$-value at step $k$ to a sampling problem: if we can sample values of $Y$ from its null distribution given $M_{k-1}(Y), U_{k-1}(Y)$, and compute the test statistic $T_k(Y)$ for each sample, we can numerically approximate the $p$-value to whatever precision we desire.

In many cases it is easy to sample $Y$ given $U_{k-1}$, in which case computational difficulties only arise insofar as the selection event $\{M_{k-1}(Y)=M\}$ is difficult to condition on, or the test statistic $T_k(Y)$ is difficult to evaluate. Due to the generality in which we have posed the problem, we cannot give a general prescription for how to carry out this sampling. However, we give some details for the important case of linear regression in Appendix~\ref{sec:linRegComputation}.

\subsection{Constraints for the Greedy Likelihood and Next-Entry Tests}\label{sec:constraints}

For a completely general selection algorithm, we might expect severe computational difficulties when we attempt to condition on the null model at stage $k$, $M_{k-1}(Y)$. For example, in forward stepwise regression, $M_{k-1}$ is the cumulative result of $k-1$ steps of greedy variable selection. At each step $s<k$, the selection of variable $j_s$ is a result of comparing its $t$-statistic with the $t$-statistics of the other $p-k$ inactive variables. Based on this logic, we might expect to accumulate at least $O(p)$ constraints at each step, resulting in $O(kp)$ constraints after $k$ steps \citep{tibshirani2014exact}.

Remarkably, however, in both forward-stepwise and lasso linear regression, the computational problem actually becomes easier when we move to the sequential setting considered here. We now show that for both forward-stepwise regression and lasso linear regression --- and their respective generalizations to exponential family models --- conditioning on the event $\{E_{k-1}=E, U_k=u\}$, for any $k, E,$ and $u$, only introduces $2(p-k)$ linear inequality constraints, amounting to a rectangular constraint for the sufficient statistics of excluded variables. This increases the speed of both accept/reject and hit-and-run sampling for our problem. We begin by proving the result for forward stepwise linear regression, followed by a more general result for the forward stepwise likelihood path in exponential families. To avoid trivialities, we assume throughout that there are almost surely no ties. This is true, for example, in linear regression if $X$ is in general position. 

\begin{theorem}
  Assume the model path $M_{0:d}$ is obtained by forward stepwise 
  linear regression. Then, for a candidate active set $E$ of size $k$, 
  the set $A = \{E_k = E, \;X_E'Y = u\}$ is characterized 
  exactly by the constraints $X_E'Y=u$ and
  \[
  v_j^-(E,u) \leq X_j'Y \leq v_j^+(E,u), \quad\forall j \notin E,
  \]
  for $v_j^-$ and $v_j^+$ given explicitly in (\ref{eq:vMinus_FS}--\ref{eq:vPlus_FS}).  Thus, $A$ corresponds exactly to a set with $2(p-k)$ linear inequality constraints and $k$ linear equality constraints on $Y$.
\end{theorem}

\begin{proof}
  At stage $i < k$, the next variable to enter is the one with maximal correlation with the residual vector, i.e.,
  \[
  j_{i+1} = \argmax_{j \notin E_i} 
  \frac{\left| X_j ' (Y - X_{E_i}\hat\beta^i) \right|}
  {\|\proj_{E_i}^\perp X_j\|_2}
  \]
  Because forward stepwise regression satisfies the SSP, we know the entire path of fitted models up to step $k$ once we condition on the $k$th active set $E_k$ and its sufficient statistics \smash{$X_{E_k}'Y$}. On the set $\{E_k=E\}$, all of the quantities $j_{i+1}$, $E_i$, $\hat\beta^i$, and \smash{$\proj_{E_i}^\perp$} depend only on $X_E'Y$. For brevity, write 
  \[
  C_i^* = \max_{j \notin E_i} 
  \frac{\left| X_j ' (Y - X_{E_i}\hat\beta^i) \right|}
  {\|\proj_{E_i}^\perp X_j\|_2}.
  \]
  On $A$, $C_i^*$ is attained at $j=j_{i+1}$, the $(i+1)$st variable added.

  If $X_E'Y$ is fixed at $u$, and $j \notin E$, then the condition for $X_j$ {\em not} to enter at step $i+1 < k$ is
  \[
  \frac{\left| X_j' (Y - X_{E_i}\hat\beta^i) \right|}
  {\|\proj_{E_i}^\perp X_j\|_2} 
  \leq C_i^*,
  \]
  or equivalently,
  \begin{equation}\label{eq:noEnterBounds_FS}
    X_j' X_{E_i}\hat\beta^i -
    C_i^*\|\proj_{E_{i}}^\perp X_{j}\|_2
    \;\;\;\leq\;\;\;
    X_j'Y
    \;\;\;\leq\;\;\;
    X_j' X_{E_i}\hat\beta^i +
    C_i^*\|\proj_{E_{i}}^\perp X_{j}\|_2, 
  \end{equation}
  so the set $A$ is equivalent to 
  \eqref{eq:noEnterBounds_FS} holding
  for every $i < k$ and $j \notin E$. Nominally, this gives $2k(p-k)$
  linear inequality constraints to satisfy, but most of them are
  non-binding. On $A$, the upper and lower bounds
  in~\eqref{eq:noEnterBounds_FS}
  are all known functions of $E$ and $u$, so we can set
  \begin{align}\label{eq:vMinus_FS}
    v_j^-(E,u) &= \max_{0 \leq i < k} \;\;X_j' X_{E_i}\hat\beta^i -
    C_i^*\|\proj_{E_{i}}^\perp X_{j}\|_2, \\
    \label{eq:vPlus_FS}
    v_j^+(E,u) &= \min_{0 \leq i < k} \;\;X_j' X_{E_i}\hat\beta^i +
    C_i^*\|\proj_{E_{i}}^\perp X_{j}\|_2.
  \end{align}  
\end{proof}

As we see below, the same result holds for the much larger class of forward stepwise likelihood paths in exponential family models.

\begin{theorem}
  Assume that $M_\infty$ is an exponential family model
  of the form
  \[
  Y \sim \exp\{ \theta'U(y) - \psi(\theta) \}\,d\nu(y),
  \]
  with $\Theta \sub \R^p$ convex, and assume that $M_{0:d}$ is the forward stepwise likelihood path with
\begin{equation}
j_k = \argmax_j \;\sup \left\{\ell(\theta; Y):\; \theta\in\Theta(E_{k-1} \cup \{j\})\right\}, \quad \text{ and } \quad
E_k = E_{k-1} \cup \{j_k\}.
\end{equation}
  Then, for a candidate active set $E$ of size $k$, the set $A = \{E_k = E, \;U_E = u\}$ is characterized exactly by the constraints $U_E = u$ and
  \[
  v_j^-(E,u) \leq U_j(Y) \leq v_j^+(E,u), \quad\forall j \notin E,
  \]
  for $v_j^-$ and $v_j^+$ given in (\ref{eq:vMinus_FSL}--\ref{eq:vPlus_FSL}). Thus, $A$ corresponds exactly to a set with $2(p-k)$ linear inequality constraints and $k$ linear equality constraints on $U(Y)$.
\end{theorem}

\begin{proof}
  For a set $E\sub \{1,\ldots,p\}$ define
  \[
  \ell_E^*(U(Y)) 
  = \sup\left\{\theta'U(Y) - \psi(\theta) :\; \theta\in \Theta(E) \right\}
  = \sup\left\{\theta_E'U_E(Y) - \psi_E(\theta_E):\; \theta\in\Theta(E) \right\},
  \]
  where $\psi_E:\; \R^{|E|} \to \R$ is defined such that $\psi(\theta)=\psi_E(\theta_E)$ for $\theta\in \Theta(E)$. Note that $\ell_E^*$ is the convex conjugate of the function $\psi_E$. For step $i<k$, write
  \[
  C_i^* = \max_{j \notin E_i} \;\ell_{E_i\cup \{j\}}^*(U(Y)).
  \]
  On $A$, $C_i^*$ is attained at $j=j_{i+1}\in E$, the next variable to enter. For $j\notin E$, the condition for $j$ {\em not} to enter at step $i+1$ is \smash{$\ell_{E_i\cup \{j\}}^*(U(Y)) \leq C_i^*$}. Due to the SSP, on the set $\{E_k=E\}$ all of the quantities $j_{i+1}$, $E_i$, and $C_i^*$ are functions of $U_E(Y)$.

  Because $\ell_{E_i\cup \{j\}^*}(U)$ depends only on $U_{E_i}$ and $U_j$, the former of which is fixed at $u_{E_i}$, we can consider it as a one-dimensional convex function of $U_j$; hence any sublevel set is an interval with possibly infinite endpoints. Thus
  \begin{equation}\label{eq:noEnterBounds_FSL}
  \{U_E=u, \; \ell_{E_i\cup \{j\}}^*(U(Y)) \leq C_i^*\} 
  = \{ U_E=u, \; v_{j,i}^-(E,u) \leq U_j \leq v_{j,i}^+(E,u)\},
  \end{equation}
  where $v_{j,i}^{-}$ and $v_{j,i}^+$ are the two values of $U_j$ such that \smash{$\ell_{E_i\cup \{j\}}^*(U)=C_i^*$}, or else $\pm\infty$.

  The set $A$ is equivalent to 
  \eqref{eq:noEnterBounds_FSL} holding
  for every $i < k$ and $j \notin E$. As before, nominally, this gives $2k(p-k)$ 
  linear inequality constraints to satisfy, but most of them are
  non-binding. On $A$, the upper and lower bounds
  in~\eqref{eq:noEnterBounds_FSL}
  are all known functions of $E$ and $u$, so we can set
  \begin{align}\label{eq:vMinus_FSL}
    v_j^-(E,u) &= \max_{0 \leq i < k} \;\;v_{j,i}^-(E,u),\\
    \label{eq:vPlus_FSL}
    v_j^+(E,u) &= \min_{0 \leq i < k} \;\;v_{j,i}^-(E,u).
  \end{align}
\end{proof}

A similar result also holds for the class of $\ell_1$-regularized exponential family models.

\begin{theorem}
  Assume that $M_\infty$ is an exponential family model
  of the form
  \[
  Y \sim \exp\{ \theta'U(y) - \psi(\theta) \}\,d\nu(y),
  \]
  with $\Theta \sub \R^p$ convex, and assume
  that the model path $M_{0:d}$ is given 
  by the ever-active set for the $\ell_1$-penalized problem
  \begin{equation}\label{eq:regProblem}
  \hat\theta^\lambda = \argmin_{\theta\in \Theta} \;
  -\ell(\theta) + \lambda\|\theta\|_1,
  \end{equation}
  for $\lambda\in \Lambda \sub [0,\infty)$.
  Then, for a candidate active set $E$ of size $k$, 
  the set $A = \{E_k = E, \;U_E = u\}$ is characterized 
  exactly by the constraints $U_E = u$ and
  \[
  v_j^-(E,u) \leq U_j(Y) \leq v_j^+(E,u), \quad\forall j \notin E,
  \]
  for $v_j^-$ and $v_j^+$ given explicitly in
  (\ref{eq:vMinus_L1}--\ref{eq:vPlus_L1}).
  Thus, $A$ corresponds exactly to 
  a set with $2(p-k)$ linear inequality constraints and $k$
  linear equality constraints on $U(Y)$.
\end{theorem}

\begin{proof}
  Note that, because exponential family log-likelihoods are concave in their natural parameters, the problem in~\eqref{eq:regProblem} is convex. Again, once we know the sufficient statistics $U_E(Y)$ for model $k$, and that $E_k=E$, we know that
  \[
  \hat\theta^\lambda = \hat\theta^{(E,\lambda)}
  \]
  for every $\lambda \geq \lambda_k$, so we know the entire path of fits up to and including $\lambda_k$. But then, excluding variable $j \notin E$ at Lagrange parameter $\lambda$ is equivalent to
  \[
  \lambda \geq 
  \left| \pardd{\ell(\hat\theta^\lambda)}{\theta_j} \right|
  = \left|U_j - \E_{\hat\theta^\lambda}[U_j]\right|.
  \]
  On $A$, for $\lambda \geq\lambda_k$, 
  \smash{$\E_{\hat\theta^\lambda}[U_j]$} is a known function of $E$ and $u$,
  so we can set
  \begin{align}\label{eq:vMinus_L1}
    v_j^-(E,u) &= \sup_{\lambda \geq \lambda_k} \;\; 
    \E_{\hat\theta^\lambda}[U_j] - \lambda, \\
    \label{eq:vPlus_L1}
    v_j^+(E,u) &= \inf_{\lambda \geq \lambda_k} \;\;
    \E_{\hat\theta^\lambda}[U_j] + \lambda.
  \end{align}
\end{proof}

\subsection{Computing the Next-Entry $p$-Values}\label{sec:next-entry-comp}

For the next entry test, the test statistic $\lambda_k(Y)$ is a function of the regularization path. A na\"{i}ve algorithm would simply recompute the entire path for each new sample $Y^*$ from the conditional null distribution in~\eqref{eq:genDist}. However, if we are a bit more clever, we can avoid ever recomputing the path.

Let $A$ denote the event $\{E_{k-1}=E, U_{k-1}=u\}$, and let $Y^*$ represent a sample from $F$, independent of $Y$. On $A$, the next-entry $p$-value $p_k$ may be written as
\begin{align}
  p_k(Y) 
  &= \P(\lambda_k(Y^*) > \lambda_k(Y) \mid E_{k-1}(Y^*)=E, M_{k-1}(Y^*)=u)\\
  &= 1 - \P\left( \bigcup_{\lambda > \lambda_k(Y)} \{\hat\theta^\lambda(Y^*) = \hat\theta^{(E, \lambda)}(Y^*)\} \mid E_{k-1}(Y^*)=E, U_{k-1}(Y^*)=u\right) \\
  &= 1 - \P\left( \bigcup_{\lambda > \lambda_k(Y)} \hat\theta^\lambda(Y^*) = \hat\theta^{\lambda}(Y) \mid E_{k-1}(Y^*)=E, U_{k-1}(Y^*)=u\right),
\end{align}
where the third equality is justified by the fact that $\hat\theta^{(E, \lambda)}$ depends only on $U_E$, and $\hat\theta^{(E,\lambda)}(Y)=\hat\theta^\lambda(Y)$ for $\lambda > \lambda_k(Y)$. As a result, we do {\em not} need to compute the whole path for each new $Y^*$; instead, we only need to compute the regularization path for $Y$, and then check for each $\lambda > \lambda_k(Y)$ whether $\hat\theta^\lambda(Y)$ is also optimal for $Y^*$. If $-\ell$ and $P_\lambda$ are both convex, then, we merely need to verify the Karush-Kuhn-Tucker (KKT) conditions.

\paragraph{Remark} If $\lambda_k$ has a discrete distribution --- as it would, for example, if $\Lambda$ is a discrete grid --- then we can either accept that $p_k$ is conservative, or randomize $p_k$ so it is exactly uniform under the null. For a randomized $p_k$, we also need the probability $\lambda_k(Y^*)$ {\em equals} $\lambda_k(Y)$, which requires computing $\hat\theta^{(E, \lambda_k)}(Y)$ and checking whether $\hat\theta^{(E, )}$.

\section{Simulation: Sparse Linear Regression}
\label{sec:sparseReg}

Here we compare several model selection procedures in simulation. We generate data from a linear regression model with $n=100$ observations and $p=40$ variables. The design matrix $X\in\R^{n\times p}$ is a random Gaussian design with pairwise correlations of 0.3 between predictor variables.

The columns of $X$ are normalized to have length $\|X_j\|=1$, and we simulate $Y \sim \cN(X\beta,I_n)$, using a seven-sparse model with signal-to-noise ratio 5:
\[
\beta_j = \left\{\begin{matrix}5 & j = 1,\ldots,7\\ 0 &
    j = 7,\ldots 100 \end{matrix}\right. .
\]
We use known $\sigma^2=1$, so that we can compare the saturated-model test with the selected-model test. For our selection algorithm, we use the entire forward-stepwise path, for all 40 steps. 

\subsection{Single-Step $p$-Values}

For each step we compute saturated-model $p$-values, max-$z$ $p$-values, and nominal (unadjusted) $p$-values, conditioning on the signs of the active variables as proposed by \citet{lee2013exact}. Figure~\ref{fig:simulation_null_false} shows the power of all three tests for each of the first eight steps, conditional on the event that the null hypothesis is false. It is clear from Figure~\ref{fig:simulation_null_false} that the selected-model $p$-values are far more powerful than the saturated-model $p$-values, especially for the first five steps, where they have near-perfect power. The nominal $p$-values are also quite powerful, but they do not have the correct level.

\begin{figure}[h]
  \centering
  \includegraphics[width=1\textwidth]{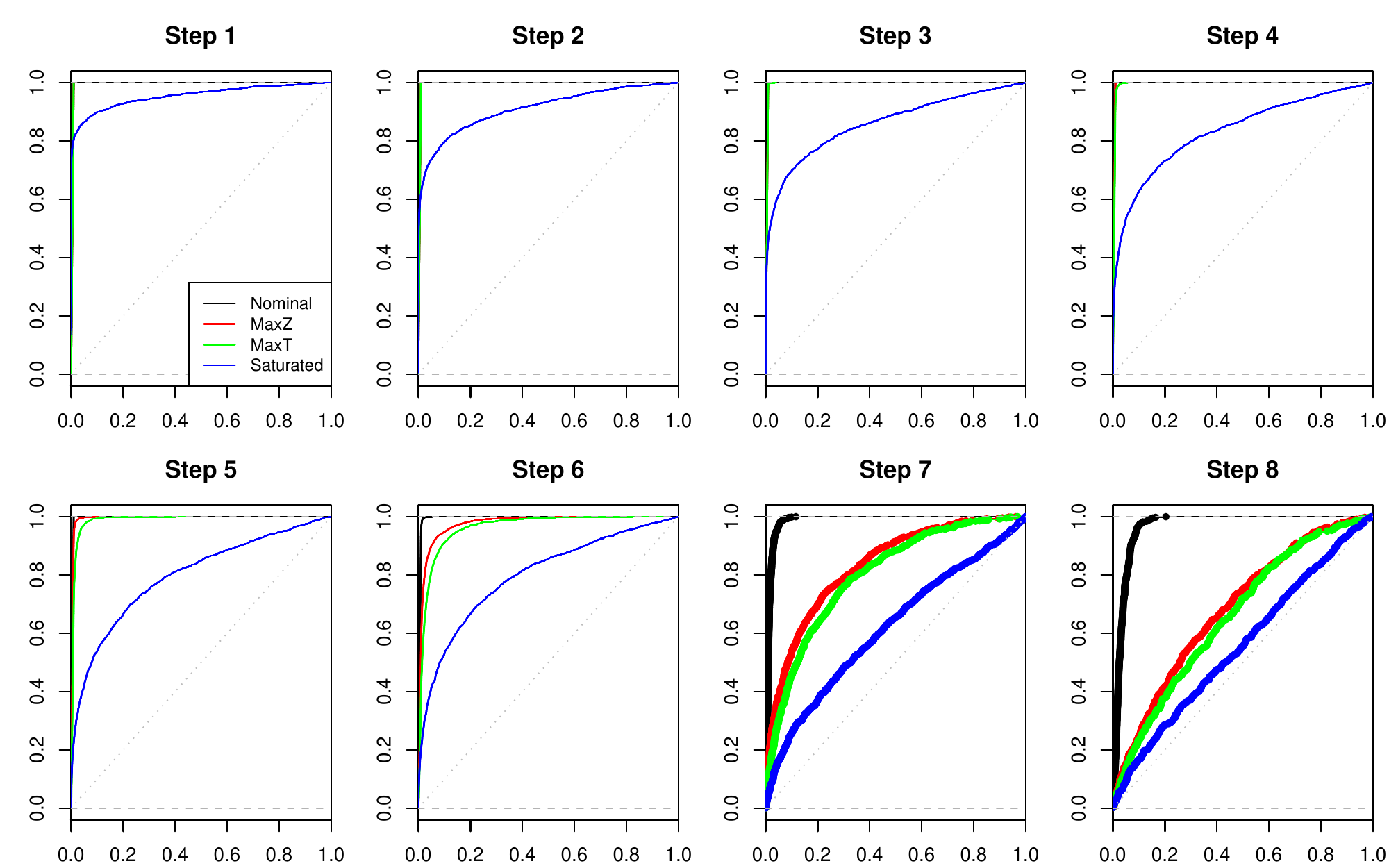}
  \caption{  CDFs of nominal (black), max-$z$ (red), max-$t$ (green) and  saturated-model (blue) $p$-values in the simulation of Section~\ref{sec:sparseReg}, conditional on testing a false null hypothesis at step $k$. The max-$t$ approaches are much more powerful than the saturated-model test. The nominal test appears to be powerful, 
  but is not $U[0,1]$ under the null, as shown below in Figure \ref{fig:simulation_null_true}.}
  \label{fig:simulation_null_false}
\end{figure}

\begin{figure}[h]
  \centering
  \includegraphics[width=1\textwidth]{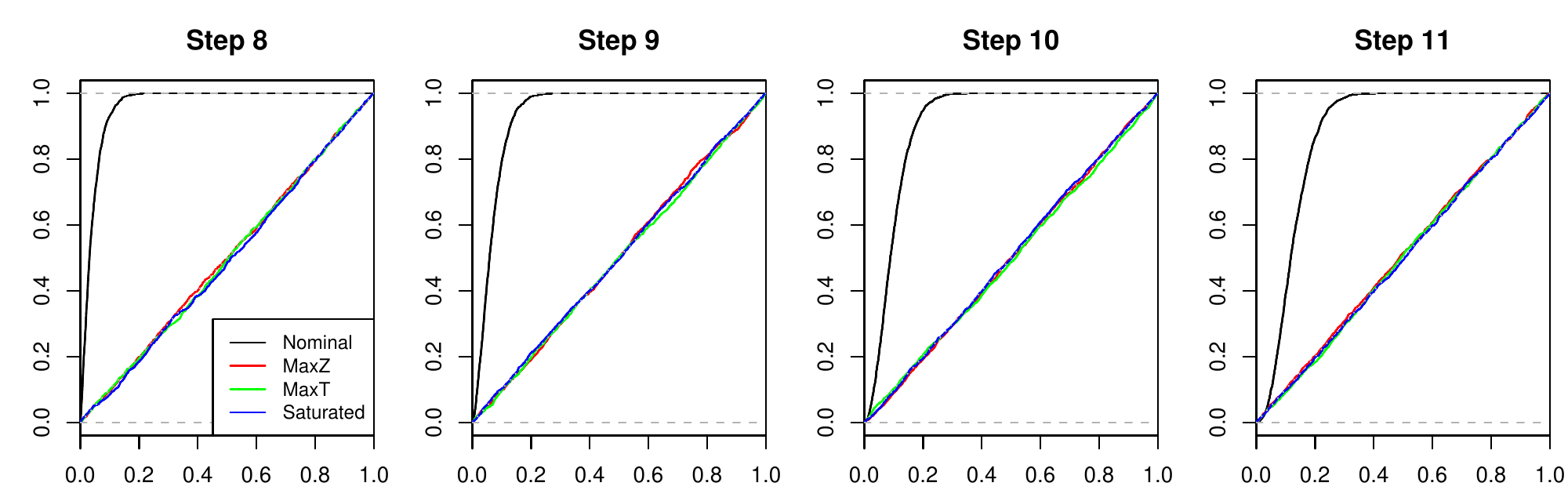}
  \caption{  CDFs of nominal (black), max-$z$ (red), max-$t$ (green) and  saturated-model (blue) $p$-values in the simulation of Section~\ref{sec:sparseReg}, conditional on testing a true null hypothesis at step $k$.  The nominal test is badly anti-conservative, while all of the other methods  show uniform $p$-values as desired.}
  \label{fig:simulation_null_true}
\end{figure}

\begin{figure}[h]
  \centering
  \includegraphics[width=1\textwidth]{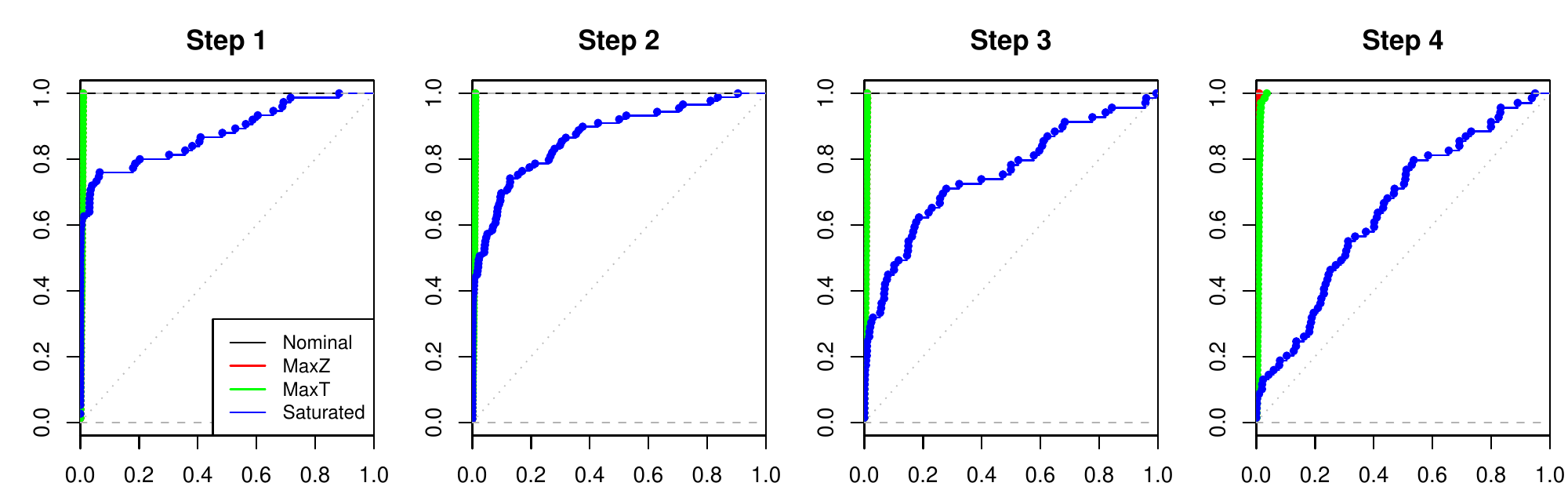}
  \caption{ CDFs of nominal (black), max-$z$ (red), max-$t$ $\sigma$ (green) and  saturated-model (blue) $p$-values in the simulation of Section~\ref{sec:sparseReg}, conditional on the event that the variable added at step $k$ is a noise variable in the full model. Here, none of the methods produce uniform $p$-values. The null hypothesis is false in most cases and so --- in our model-centric point of view --- rejection is the desired outcome.}
  \label{fig:simulation_noise_var}
\end{figure}

Figure~\ref{fig:simulation_null_true} shows the distribution of $p_k$ for steps eight to eleven, given that the null hypothesis tested at step $k$ is correct (i.e., that $k_0\leq k$). Because the correct model is seven-sparse, $k=8$ is the first index for which the null model can possibly be true. All but the nominal $p$-values are uniform by construction, and the nominal $p$-values are highly anti-conservative, as expected.

Finally, as a warning against misinterpretation of our method, we include Figure~\ref{fig:simulation_noise_var} showing the first four $p$-values for each method, conditional on event that the variable added at the given step is a noise variable in the full model. Now, none of the $p$-values are uniform. This is not an error in our implementation of the method, but rather a consequence of our ``model-centric'' point of view. If we try to add a noise variable to the model before we have included all the signal variables, then we are testing a false null hypothesis. The test rejects because there is much more signal to find, and as such, the null model is false.  

\subsection{Model-Selection Performance}

If we combine the saturated-model or selected-model $p$-values with one of our three stopping rules, we can evaluate the model-selection performance of each method. Table~\ref{tab:stopping05} reports several model-selection metrics for each combination of test and stopping rule carried out at level $\alpha=0.05$, while Table~\ref{tab:stopping20} reports the same results for the more liberal $\alpha=0.2$. The max-$t$ test, which we could use in place of the max-$z$ test if $\sigma^2$ were unknown, is included for comparison. The max-$z$-identify test was also included in the simulation, but not shown in Tables~\ref{tab:stopping05} and \ref{tab:stopping20} because its performance on all metrics was virtually identical to the max-$z$ test.

The metrics are $\P(\hk \geq k_0)$, the probability of selecting a correct model; the (model-wise) FWER; the FWER conditional on $\hk \geq k_0$, labeled cFWER; the model-wise FDR and variable-wise $\text{FDR}^{\text{full}}$; and the average number of full-model signal variables selected, labeled $\E[S^{\text{full}}]$ (where $S^{\text{full}} = R-V^{\text{full}}$). Note that $\text{FDR}^{\text{full}}$ is not explicitly controlled by any of the methods other than knockoff+, but we might nevertheless hope to perform reasonably well. Also, note that the model-wise FWER and FDR and not defined for knockoffs, since it does not select variables sequentially, and never formally tests any model.

\begin{table}[ht]
\centering

\newcommand{\guarantee}[1]{{\bf #1}}
{\small 
\begin{tabular}{|l l|cccccc|}
\hline
{} & {} &  $\P(\hk \geq k_0)$ &  $\text{FWER}$ &  $\text{cFWER}$ &  $\text{FDR}$ &  $\text{FDR}^{\text{full}}$ &  $\E[S^{\text{full}}]$ \\ \hline
Nominal & BasicStop & 0.83 & 0.58 & 0.70 & 0.120 & 0.198 & 6.81 \\ 
Nominal & ForwardStop & 0.93 & 0.92 & 0.99 & 0.387 & 0.471 & 6.93 \\ 
Max-z & BasicStop & 0.37 & \guarantee{0.02} & \guarantee{0.06} & \guarantee{0.003} & 0.045 & 6.15 \\ 
Max-z & ForwardStop & 0.66 & 0.19 & 0.29 & \guarantee{0.028} & 0.088 & 6.59 \\ 
Max-t & BasicStop & 0.32 & \guarantee{0.02} & \guarantee{0.06} & \guarantee{0.003} & 0.043 & 5.97 \\ 
Max-t & ForwardStop & 0.63 & 0.16 & 0.26 & \guarantee{0.023} & 0.080 & 6.53 \\ 
Saturated & BasicStop & 0.00 & 0.00 & 0.00 & 0.000 & 0.024 & 2.04 \\ 
Saturated & ForwardStop & 0.04 & 0.00 & 0.03 & 0.000 & 0.027 & 2.82 \\ 
Knockoff & & 0.29 & --- & --- & --- & 0.107 & 4.31 \\ 
Knockoff+ & & 0.00 & --- & --- & --- & \guarantee{0.000} & 0.00 \\  \hline
\end{tabular}}
\caption[tab:stopping]{Model-selection performance of various stopping rules applied to simulated data with 7 strong signals,  at level $\alpha=0.05$. Values theoretically guaranteed to be less than $\alpha$ are in \guarantee{bold} type. The largest standard error in each of the six columns is 0.01, 0.01, 0.02, 0.004, 0.003, and 0.05.}
\label{tab:stopping05}
\end{table}

\begin{table}[ht]
\centering

\newcommand{\guarantee}[1]{{\bf #1}}
{\small 
\begin{tabular}{|l l|cccccc|}
\hline
{} & {} &  $\P(\hk \geq k_0)$ &  $\text{FWER}$ &  $\text{cFWER}$ &  $\text{FDR}$ &  $\text{FDR}^{\text{full}}$ &  $\E[S^{\text{full}}]$ \\ \hline
Nominal & BasicStop & 0.94 & 0.93 & 0.99 & 0.415 & 0.498 & 6.94 \\ 
Nominal & ForwardStop & 0.98 & 0.98 & 1.00 & 0.640 & 0.704 & 6.98 \\ 
Max-z & BasicStop & 0.61 & \guarantee{0.12} & \guarantee{0.20} & \guarantee{0.018} & 0.073 & 6.53 \\ 
Max-z & ForwardStop & 0.86 & 0.64 & 0.75 & \guarantee{0.152} & 0.231 & 6.84 \\ 
Max-t & BasicStop & 0.58 & \guarantee{0.11} & \guarantee{0.19} & \guarantee{0.016} & 0.069 & 6.46 \\ 
Max-t & ForwardStop & 0.84 & 0.62 & 0.73 & \guarantee{0.140} & 0.219 & 6.82 \\ 
Saturated & BasicStop & 0.06 & 0.01 & 0.13 & 0.001 & 0.031 & 3.13 \\ 
Saturated & ForwardStop & 0.39 & 0.17 & 0.44 & 0.028 & 0.074 & 4.98 \\ 
Knockoff & & 0.59 & --- & --- & --- & 0.230 & 5.72 \\ 
Knockoff+ & & 0.36 & --- & --- & --- & \guarantee{0.136} & 3.76 \\  \hline
\end{tabular}}
\caption[tab:stopping]{Model-selection performance of various stopping rules applied to simulated data with 7 strong signals, at level $\alpha=0.2$. Values theoretically guaranteed to be less than $\alpha$ are in \guarantee{bold} type. The largest standard error in each of the six columns is 0.01, 0.01, 0.03, 0.004, 0.004, and 0.07.}
\label{tab:stopping20}
\end{table}

We see that BasicStop and ForwardStop control model-wise FWER and FDR, respectively, using $p$-values from max-$z$ and max-$t$, as predicted by the theory, and the nominal $p$-values do not lead to control of FWER or FDR, as expected due to their anti-conservatism. Although not guaranteed by the theory, the $p$-values from the saturated model do control model-wise FDR in this example.  The max-$z$ method is more powerful, showing higher probabilities of selecting a correct model, especially when using ForwardStop. The max-$t$ method is also powerful but slightly less so than the max-$z$, owing to the additional nuisance variable $\sigma^2$.

The knockoff method does not yield variable-wise FDR control. While the more conservative knockoff+ method does achieve variable-wise FDR control, this comes at a cost of very low power. Because of the discrete nature of knockoff and knockoff+, these two methods might perform better in an example with a larger number of signal variables to find.

\section{Further Examples and Extensions}\label{sec:further-examples}

Although we have focused our discussion on the case of linear regression --- a very common and important application of sequential model selection --- our theory is quite general and applies to many other parametric and nonparametric models. To give a sense of the broader applicability, we now discuss one parametric and one nonparametric example, each of which fit nicely into our theory.

\subsection{Decision Trees}
In this section we discuss another application of adaptive sequential inference, to binary decision trees. Specifically, the classification and regression tree  (CART)  method builds a decision tree in a top down manner, finding the best available split at each node \citep{breiman1984classification}.  When considering splitting two daughter nodes, the generic CART algorithm splits the nodes in an arbitrary order, only stopping when a node reaches a minimum size.  We instead consider a  special version of CART: the ``best first'' method, which orders splits by their achieved reduction in loss. This procedure is used, for example in the R package for gradient boosting, called {\tt gbm}. The best-first algorithm is thus a sequential selection  procedure, to which we can apply the results in this paper. Hence we can obtain independent, sequential $p$-values for each successive split in the decision tree. 

Given predictor matrix $X \in \R^{n\times p}$ and binary response vector $Y\in \{0,1\}^n$, we begin with the intercept-only model $\logit \P(Y_i \mid X_i = x) \equiv \beta_0$. At the first step we choose a splitting variable $j_1$ and split point $t_1$ which is a midpoint between two observed values of variable $j$, leading to a model of the form
\begin{equation*}
\logit \,\P(Y_i \mid X_i=x) = \beta_0 + \beta_{1} 1\{x_{j_1} > t_1\}.
\end{equation*}
While there are various ways to choose among potential variables and split points, we assume for the sake of simplicity that $(j_1,t_1)$ are chosen to maximize binomial likelihood (or equivalently, to minimize binomial deviance). At the second stage, we then split one of the two daughter nodes, choosing the split that gives the greatest increase in likelihood. If (say) we split the first node using predictor $j_2$ at split point $t_2$, it leads to the interaction model of the form
\begin{equation}
\logit \,\P(Y_i \mid X_i=x) = \beta_0 + \beta_{1} 1\{x_{j_1} > t_1\} + \beta_2 1\{x_{j_1} \leq t_1, x_{j_2} > t_2\}.
\label{eqn:intmodel}
\end{equation}
After $k$ steps we obtain a logistic regression model with $k$ features, each of which is an interaction of step functions. The model $M_k$ is thus an exponential family whose complete sufficient statistics are the sums of $Y_i$ in each leaf in the tree.

The best-first search for split points can be expressed as a modified form of greedy forward stepwise logistic regression in which the set of candidate features changes at each step, depending on what features have already been included. Nevertheless, the proof of Proposition~\ref{prop:forwardSSP} goes through essentially without modification to show that this algorithm satisfies the SSP, and the greedy likelihood ratio test gives $p$-values that are independent on nulls. Since the test statistic is discrete, we need to randomize if we want truly exact $p$-values. 

This development suggests a possible modification to the CART algorithm in which the tree is ``pruned'' using a rule like ForwardStop applied. It could be interesting to investigate whether cross-validation using the target FDR $\alpha$ as a tuning parameter might yield better results in predictive performance, compared to (say) using the tree depth or minimum leaf size as a tuning parameter.

\subsection{Nonparametric Changepoint Detection}\label{sec:nonpar}

Consider a problem in which we observe a time series $(Y_1,\ldots,Y_T)\in \cY^T$, and wish to detect times $t_1,\ldots,t_k$ where the distribution of $Y_t$ changes. For example, given a historical document, we may wish to investigate whether the entire document was written by the same author or different sections were written by different authors, by quantitatively comparing the writing style across different sections. For example, $Y_t$ could could represent the usage distribution in section $t$ over ``context-free'' words such as articles, conjunctions, and prepositions \citep[see e.g.,][]{chen2015graph}.

For simplicity we assume independence between values of $t$ and model the distribution as piecewise constant between $k$ unknown changepoints. For an active set of $k$ changepoints $E\sub \{1,\ldots, T-1\}$, let $t(j,E)$ denote the $j$th smallest element of $E$, and let $t(0,E)=0$ and $t(k+1, E)=T$. Then, the nonparametric model with changepoints $E$ is given by
\begin{equation}\label{eq:nonpar-model}
  M(E):\; Y_t \simind F_j, \quad \text{ for } \;\;t(j,E) < t \leq t(j+1,E).
\end{equation}
For $t_2>t_1$ let $V(Y;\, t_1, t_2)\in \cY^{t_2-t_1}$ denote the order statistics of $(Y_{t_1+1},\ldots, Y_{t_2})$. Then the complete sufficient statistic for $M(E)$ is
\[
U(Y; \, E) = 
\bigg(
V(Y; t(0, E), t(1, E)), \ldots, V(Y;\, t(k, E), t(k+1,E))
\bigg).
\]
Under $M(E)$, resampling from $Y$ given $U$ amounts to randomly permuting $(Y_{t(j,E)+1}, \ldots, Y_{t(j+1,E)})$ independently for each $j=0,\ldots,k$.

For $s \notin E$, testing $M(E)$ against $M(E \cup \{s\})$ amounts to a two-sample test comparing the observations coming immediately before $s$ against the observations coming immediately after $s$. Specifically, if $t(j, E) < s < t(j+1, E)$, let $W(Y; \, E, s)$ denote any two-sample test statistic that is measurable with respect to the samples $V(Y; t(j, E), s)$ and $V(Y; s, t(j+1, E))$, and large when there is strong evidence that the two samples come from different distributions. 

We can build a model in a greedy fashion beginning with $E_0 = \emptyset$ and then at step $k=1,\ldots,d$ setting
\begin{equation*}
  t_k = \argmax_t \; W(Y; \, E_{k-1}, t), 
  \quad \text{ and } \quad E_k = E_{k-1} \cup \{t_k\}.
\end{equation*}
This path algorithm satisfies the SSP because $W(Y; \, E_{i}, t_j)$ is measurable with respect to $U_k=U(Y; \, E_k)$ for every $i<j\leq k$. Thus, by inspecting $U_k$ we can determine in what order the first $k$ changepoints were added. 

By analogy to the greedy likelihood ratio statistic, a natural choice of test statistic at step $k$ is
\[
W_k^*(Y) = \max_t \; W(Y; \, E_{k-1}, t),
\]
which is measurable with respect to $\sF(E_k, U_k)$ by construction. Thus, if we base the step-$k$ $p$-value $p_k$ on the law
\begin{equation}\label{eq:condPerm}
\cL\left(W_k^* \mid E_{k-1}, U_{k-1}\right),
\end{equation}
randomizing to obtain an exactly uniform $p$-value under $M(E_{k-1})$, we will obtain an exact $\sF_k$-measurable $p$-value. As a result, the $p$-values will be independent on nulls per Theorem~\ref{thm:suffCond}. Note that sampling from the law in~\eqref{eq:condPerm} can be carried out by permuting subsequences $(Y_{t(j,E)+1}, \ldots, Y_{t(j+1,E)})$ and accepting the permutation only if it gives the same $E_{k-1}$ as we get in the original data. When this simple accept/reject algorithm is impractical we may need to resort to an MCMC strategy. We do not pursue these computational matters here.

\section{Discussion}
\label{sec:discussion}

A common proverb in statistics states that ``essentially all models are wrong, but some are useful'' \citep{box1987empirical}. In essence, a statistical model is useful if it is large enough to capture the most important features of the data, but still small enough that inference procedures can achieve adequate power and precision. Apart from theoretical considerations, the only way to know whether a model is large enough is to test it using available data.

Although model-checking is commonly recommended to practitioners as an important step in data analysis, it formally invalidates any inferences that are performed with respect to the model selected. Our work takes a step in the direction of reconciling that contradiction, but there are important questions left to be resolved. In particular: which sorts of model misspecification pose the greatest threat to our inferential conclusions, and how powerful are our tests against these troublesome sources of misspecification? 

In future work we also plan to address the matter of performing selective inference for parameters of the model actually selected. In principle, we already know how to do this using the framework of \citet{fithian2014optimal} --- we simply condition on the event that $M_{\hk}$ is selected and perform inference with respect to \smash{$M_{\hk}$} --- but the task may be computationally awkward in general.

Open-source code is available at this article's github repository at:\\ \texttt{github.com/wfithian/sequential-selection}.

R and Python code are available to implement proposals in this paper are available at:\\ \texttt{github.com/selective-inference}.
They will also be added to an upcoming version of the  R package  \texttt{selectiveInference} on CRAN.

\section*{Acknowledgments}

The authors are grateful for stimulating and informative conversations with Stefan Wager, Lucas Janson, Trevor Hastie, Yoav Benjamini, and Larry Brown. William Fithian was supported in part by the Gerald J. Lieberman Fellowship. Robert Tibshirani was supported by NSF grant DMS-9971405 and NIH grant N01-HV-28183. Ryan Tibshirani was supported by NSF grant DMS-1309174. Jonathan Taylor was supported by NSF grant DMS-1208857.

\bibliographystyle{plainnat}
\bibliography{biblio}

\newpage

\begin{appendix}

\section{Proof of Proposition~\ref{prop:counterexample}}

In this section we prove Proposition~\ref{prop:counterexample}, restated below for convenience.

\counterex*

\begin{proof}
Denote the events where we choose the two possible orderings as $A_{123}$ and $A_{213}$, respectively. Because the variables are orthogonal, the max-$z$ statistic for any model $M(E)$ is 
\[
|z_E^*|=\max_{j\notin E} |Y_j|.
\]

At each step we use the selective max-$z$ test conditioning on the current null model, where $p_k$ is small when $|z_{E_{k-1}}^*|$ is large compared to the law
\[
\cL(|z_{E_{k-1}}^*| \,\mid\, E_{k-1}, Y_{E_{k-1}}).
\]

Note that whichever of $A_{123}$ and $A_{213}$ occurs, $E_2=\{1,2\}$ and $|z_{E_2}^*|=|Y_3|$. Thus, $E_2$ is nonrandom and $p_3=2\Phi(-|Y_3|)$, the usual two-sided normal $p$-value based on $Y_3$. Finally, note that $k_0=2$ on $A_{123}$ and $k_0=1$ on $A_{213}$. 

We have
\begin{align*}
  FWER 
  &= \P(A_{123})\P(p_1<\alpha, p_2<\alpha, p_3<\alpha \mid A_{123})
  + \P(A_{213})\P(p_1<\alpha, p_2<\alpha \mid A_{213})\\
  &= \alpha \P(p_1<\alpha, p_2<\alpha \mid A_{123})
  + (1-\alpha) \P(p_1<\alpha, p_2<\alpha \mid A_{213}),
\end{align*}
using the fact that $A_{123}$ corresponds exactly to the event that $p_3<\alpha$. As $K\to\infty$, we have $z_\emptyset^*=z_{\{1\}}^*=Y_2>K/2$ with probability tending to 1. Thus, $p_1,p_2\approx 0$ on $A_{123}$, while $p_1\approx 0$ on $A_{213}$. Continuing the calculation, then,
\begin{align*}
  \lim_{K\to\infty} FWER 
  &= \alpha\cdot 1 + (1-\alpha)\cdot \P(p_2 < \alpha \mid A_{213})\\
  &= \alpha + (1-\alpha)\alpha = 2\alpha-\alpha^2.
\end{align*}
\end{proof}

\section{Computational Details for Linear Regression}\label{sec:linRegComputation}

As discussed in Section~\ref{sec:singleStep}, we will set
\[
p_k(Y) = p_{k, E_{k-1}(Y)}(Y),
\]
where $p_{k,E}(Y)$ for each $E\sub \{1,\ldots, p\}$ is a conditional $p$-value for the {\em fixed} null model $M(E):\, Y \sim \cN(X_E\beta, \sigma^2)$, valid on the selection event $A=\{E_{k-1}(Y)=E\}$.

If $\sigma^2$ is known, then $U_{k-1}(Y) = X_{E}'Y$ on $A$, so our task is to sample under the null from
\begin{equation}\label{eq:truncNorm}
\cL(Y \mid X_E'Y, Y\in A),
\end{equation}
which is a multivariate Gaussian distribution supported on the hyperplane of dimension $n-|E|$ with $X_E'Y$ fixed at its realized value, and truncated to the event $A$. 

If $\sigma^2$ is unknown, then $U_{k-1}=(X_E'Y,\, \|Y\|^2)$ on $A$ and we must sample under the null from
\begin{equation}\label{eq:truncUnif}
\cL(Y \mid X_E'Y,\, \|Y\|^2,\, Y\in A),
\end{equation}
which is a uniform distribution supported on a sphere of dimension $n-|E|-1$ defined by fixing $X_E'Y$ and $\|Y\|^2$ at their realized values, and truncated to the event $A$. For the simulations in this article, we use a hybrid of two strategies:

\paragraph{Accept/Reject Sampler} In either of these two cases, sampling would be very easy if $A$ were the entire sample space $\R^n$; the only reason we might encounter difficulty is that $A$ could be an irregular set that is difficult to condition on. However, we can always sample from $\cL(Y \mid U)$ and then keep each new sample $Y$ if it lies in the selection event.

\paragraph{Hit-and-Run Sampler} \citet{fithian2014optimal} propose sampling from laws like that in~\eqref{eq:truncNorm} via a hit-and-run sampler for the $(n-|E|)$-dimensional multivariate Gaussian distribution: at each step, to generate $Y^{t+1}$ from $Y^t$, we choose a uniformly random unit vector $\nu$ in the hyperplane, and resample $\nu'Y$ conditional on the constraints $\proj_{\nu}^\perp Y^{t+1}=\proj_{\nu}^\perp Y^t$ and the constraint $Y^{t+1}\in A$. If $A$ is a polytope this amounts to resampling a univariate Gaussian distribution with known mean and variance, truncated to an interval defined by maximizing and minimizing $\nu'Y^{t+1}$ subject to linear constraints. For the case with $\sigma^2$ unknown, the sampler is a bit more involved; see \citet{fithian2014optimal}.

Figure \ref{fig:comparison} compares the two samplers in an example with $n=100, p=20$ and one strong signal, computing the max-$t$ $p$-value at Step 3.
The ``exact'' answer was determined from bootstrap sampling with $20,000$ replications.
The Figures show the bias, variance and mean squared error (MSE) of the estimates from accept/reject and hit and run, using between 1000 and 8000 samples, averaged
over 10 different starting seeds.
\begin{figure}[htp]
\centering
  \includegraphics[width=\textwidth]{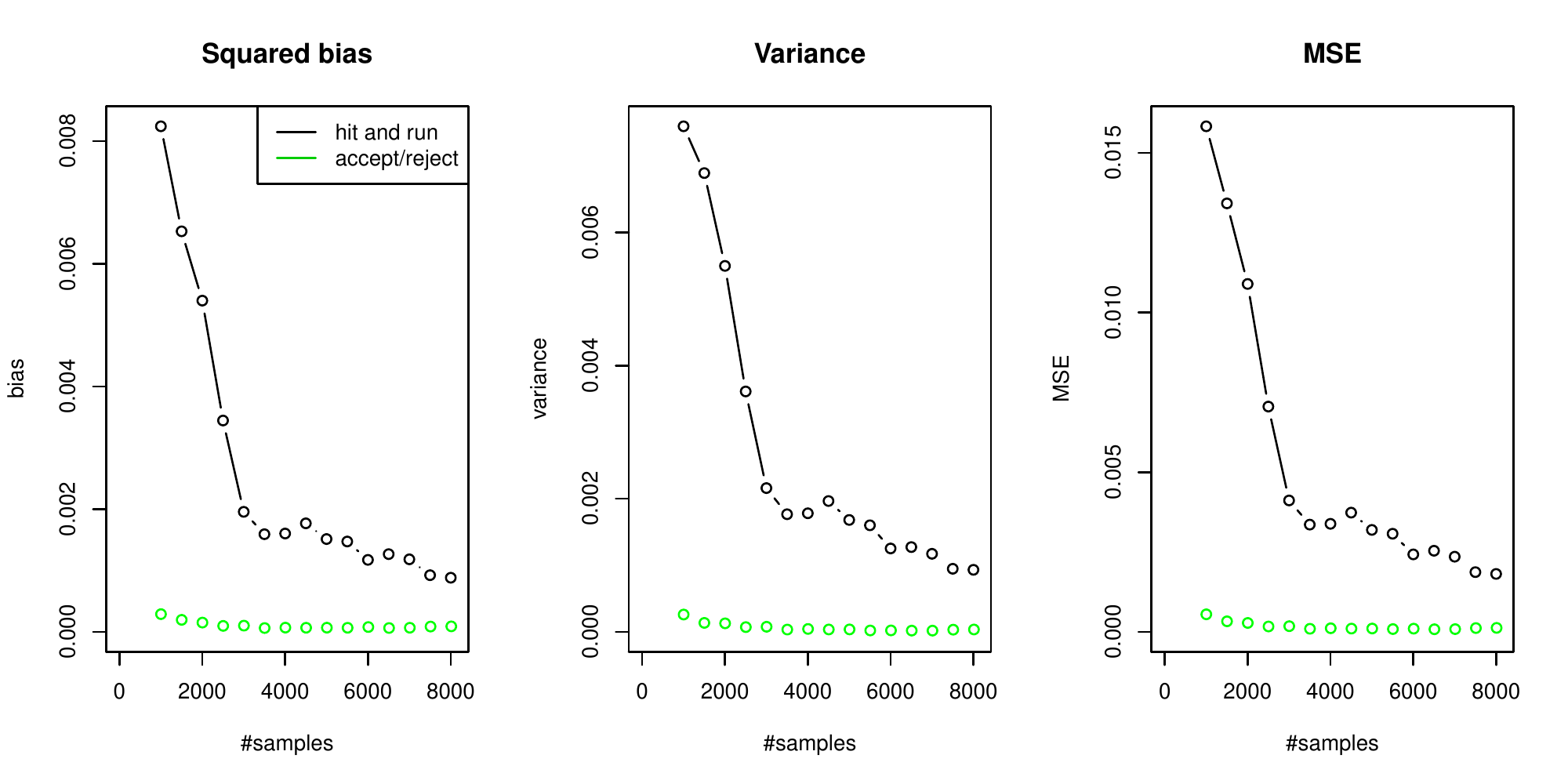}
  \caption{ Simulated example with $n= 100, p=20$ and one strong signal, computing the max-$t$ $p$-value at Step 3.  Shown are bias, variance and MSE of the estimates from accept/reject and hit and run, using between 1000 and 8000 samples, averaged over 10 different starting seeds.
}
\label{fig:comparison}
\end{figure}
We see that accept/reject is far more efficient, yielding a smaller MSE with 1000 samples then hit and run gives for 8000 samples. As we move further along the path, the set $A$ becomes more and more constraining, so that we must discard the vast majority of the samples we generate. Thus we use accept/reject for as many steps of the sequence is possible,  until the point where the acceptance rate is too low for accept/reject to be practical.

In detail, we employ a hybrid computational strategy; we use accept-reject sampling, and do so as long as we obtain at least 300 (say) samples in the first set of $B$ realizations ($75,000$ in our code)  This is an effective approach until we have reached a point in the model path where all predictors have little or no signal. When the acceptance rate gets too low, we switch to a hit-and run Monte Carlo approach. We exploit the fact that each move along a random direction  has a truncated Gaussian distribution whose truncation limits are easily computed from the properties of the polyhedron. Unlike accept/reject sampling, which produces independent samples,  the hit and run method produces dependent samples. Hence we must run it for longer with an initial burn in period, and hope that it has mixed sufficiently well.

The task is made much simpler by the fact that the selection event can we written as a polytope $A=\{y: \Gamma y \geq u\}$. Hence we don't need to run the entire forward stepwise procedure on $Y^*$: instead, we pre-compute $\Gamma$ and $u$ and then check if $\Gamma Y^* \geq u$. In the case of forward stepwise and lasso regression, $\Gamma$ has $2(p-k)$ rows after $k$ steps, as we see next.

\end{appendix}

\end{document}